\documentclass[11pt,cleveref]{jmlr}
\usepackage{newtxtext}
\usepackage[varg]{newtxmath}
\usepackage{mathtools,booktabs,enumitem,microtype,array}
\allowdisplaybreaks
\setlist[itemize]{leftmargin=1.5em,itemsep=0.15em,topsep=0.3em}
\setlist[enumerate]{leftmargin=1.8em,itemsep=0.2em,topsep=0.3em}
\newtheorem{Thm}{Theorem}
\newtheorem{Lem}{Lemma}
\newtheorem{Prop}{Proposition}
\newtheorem{Cor}{Corollary}
\crefname{Thm}{Theorem}{Theorems}
\Crefname{Thm}{Theorem}{Theorems}
\crefname{Lem}{Lemma}{Lemmas}
\Crefname{Lem}{Lemma}{Lemmas}
\crefname{Prop}{Proposition}{Propositions}
\Crefname{Prop}{Proposition}{Propositions}
\crefname{Cor}{Corollary}{Corollaries}
\Crefname{Cor}{Corollary}{Corollaries}
\crefname{section}{Section}{Sections}
\crefname{subsection}{Section}{Sections}
\crefname{appendix}{Appendix}{Appendices}
\crefname{table}{Table}{Tables}
\crefname{algocf}{Algorithm}{Algorithms}
\Crefname{algocf}{Algorithm}{Algorithms}
\newcommand{\R}{\mathbb{R}}
\newcommand{\E}{\mathbb{E}}
\newcommand{\Pp}{\mathbb{P}}
\newcommand{\KL}{\operatorname{KL}}
\newcommand{\Ber}{\operatorname{Ber}}
\newcommand{\Reg}{\operatorname{Reg}}
\newcommand{\sig}{\sigma}
\newcommand{\cR}{\mathcal{R}}
\newcommand{\cE}{\mathcal{E}}

\newcommand{\tr}{\operatorname{tr}}
\newcommand{\diag}{\operatorname{diag}}

\newcommand{\eps}{\varepsilon}

\newcommand{\norm}[1]{\lVert #1\rVert}
\newcommand{\abs}[1]{\lvert #1\rvert}
\newcommand{\defeq}{\mathrel{:=}}
\newcommand{\ind}{\mathbf{1}}
\newcommand{\dent}{\mathfrak{h}}

\title[Sharp Minimax Regret for Infinite-Memory Logistic Prediction]{Sharp Minimax Regret for Infinite-Memory Logistic Prediction}
\author{\Name{Vaneet Aggarwal}\\
        \addr Purdue University}
\hypersetup{
  pdftitle={Sharp Minimax Regret for Infinite-Memory Logistic Prediction},
  pdfauthor={Vaneet Aggarwal}
}
\begin{document}
\maketitle
\thispagestyle{plain}
\begin{abstract}
We study sequential prediction of binary marks whose log-odds are a linear function of all past inputs. The inputs are independent uniform signs. The coefficient of the input observed $j$ steps earlier is unknown, but its absolute value is at most a known number $r_j$, and $\sum_j r_j\le B$. A learner is judged by its expected cumulative log loss in excess of the loss of the true conditional probabilities. We show that for every such sequence $r$ and every horizon $T$, the minimax value of this regret lies between $c\,\Gamma_T(r)/\max\{1,B^2\}$ and $\frac12\Gamma_T(r)$, where $\Gamma_T(r)=\sum_{j\le T}\log(1+(T-j+1)r_j^2)$ and $c$ is an absolute constant. The term for lag $j$ combines the size bound $r_j$ with the number $T-j+1$ of rounds in which that lag influences a prediction, and no monotonicity or other regularity of $r$ is needed. The upper bound is attained by a Bayesian mixture with independent Gaussian priors. For the lower bound, we apply the van Trees inequality at each prediction step and then bound the expected log-determinant of the random matrix whose rows are the lagged input vectors. We also characterize when $\frac12$ is the exact leading constant: for envelopes with a common bound $B$, possibly changing with $T$, the regret is asymptotic to $\frac12\Gamma_T(r)$ if and only if $\sum_{j\le T}\min\{1,(T-j+1)r_j^2\}=o(\Gamma_T(r))$. This gives, for example, $\log^2(A^2T)/(8\alpha)$ for $r_j=Ae^{-\alpha j}$, whereas for $r_j=Aj^{-s}$ with $s>1$ the regret is of order $T^{1/(2s)}$ and its ratio to $\Gamma_T(r)$ stays below $\frac12$. Finally, for a known horizon, an online Newton predictor with features scaled by $r_j$ attains regret of order $\Gamma_T(r)$ with polynomial computation per round.
\end{abstract}
\begin{keywords}
sequential prediction, minimax regret, logarithmic loss, infinite memory,
universal coding, online Newton step
\end{keywords}
\section{Introduction}\label{sec:intro}
In many sequential problems a binary event responds to an external stream of inputs, and the influence of an input fades with its age without disappearing after a fixed delay. Categorical time-series models with exogenous covariates of this kind are studied by \citet{fokianostruquet2019}. A simple instance draws each event from a logistic probability whose log-odds are a weighted sum of all past inputs. When the weights are unknown, a predictor must learn them while it predicts, and the number of weights it may need is not bounded in advance. We determine how much cumulative log loss this learning costs.

Let $U_1,U_2,\ldots$ be independent uniform signs in $\{-1,1\}$, called inputs, and let $Y_1,Y_2,\ldots$ be marks in $\{0,1\}$. After observing $U_{1:t}$ and $Y_{1:t}$, the learner predicts $Y_{t+1}$. Under the parameter $\theta=(\theta_1,\theta_2,\ldots)$,
\begin{equation}
 \Pp_\theta(Y_{t+1}=1\mid U_{1:t},Y_{1:t})
 =\sig\left(\sum_{j=1}^{t}\theta_jU_{t+1-j}\right),
 \qquad \sig(z)=\frac{1}{1+e^{-z}}. \label{eq:intro-model}
\end{equation}
We call $\theta_j$ the coefficient of lag $j$. The input history is empty at time $1$, which is why the sum in \eqref{eq:intro-model} stops at $j=t$. The coefficients are unknown, but they satisfy $\abs{\theta_j}\le r_j$ for a known nonnegative sequence $r=(r_1,r_2,\ldots)$, the \emph{envelope}, with $\sum_j r_j\le B<\infty$. If infinitely many $r_j$ are positive, the prediction at round $t+1$ can depend on every input observed so far. The inputs are exogenous: the marks do not feed back into them. The learner assigns a probability to each possible outcome and pays the log loss $-\log q(Y_{t+1})$ of the probability $q(Y_{t+1})$ it gave to the realized mark. Its regret over $T$ rounds is its expected total log loss minus the expected total log loss of the true conditional probabilities, and $\cR_T(r)$ denotes the smallest worst-case regret over the class; \Cref{sec:setup} gives the precise definitions. Regret equals the relative entropy between the joint distribution of the source and that of the learner, so $\cR_T(r)$ is also the minimax redundancy of universal coding for this class \citep{davisson1973,merhavfeder1998,haussler1997}.

A conventional bound on $\cR_T(r)$ fits the first $h$ coefficients, charges each of them about $\log T$ as for an $h$-dimensional parametric model, and adds the loss caused by ignoring the remaining lags. This charge misses two features of the problem. First, a coefficient known to lie in $[-r_j,r_j]$ with small $r_j$ has a small effect on every prediction, so it can be cheap to learn even when it cannot be estimated accurately. Second, lag $j$ affects only the predictions of $Y_{t+1}$ with $t\ge j$, that is,
\[
 n_{T,j}=(T-j+1)_+
\]
of the $T$ predictions, so late lags are observed less often than early ones. For $r_j=Aj^{-s}$ the dimension-only bound exceeds the minimax regret by the factor $(\log T)^{1-1/(2s)}$ (\Cref{sec:rate-order}). The excess is caused by the uniform charge per lag, not by truncation: a predictor that scales its features by $r_j$ attains the minimax order using a finite window of lags (\Cref{sec:algorithm}).

Our main result, \Cref{thm:main}, accounts for both features. For every envelope and every horizon $T\ge1$,
\begin{equation}
 \Gamma_T(r)=\sum_{j=1}^{T}\log(1+n_{T,j}r_j^2),
 \qquad \frac{c_*}{\max\{1,B^2\}}\Gamma_T(r)\le\cR_T(r)\le\frac12\Gamma_T(r), \label{eq:intro-result}
\end{equation}
where $c_*>0$ is an absolute constant. The envelope may have gaps and need not be monotone. Each term of $\Gamma_T(r)$ has a direct interpretation. When $n_{T,j}r_j^2\le1$, the term is comparable to $n_{T,j}r_j^2$, and by \Cref{lem:pairwise} this quantity is, up to constants, the relative entropy between the source with $\theta_j=r_j$ and the source with $\theta_j=0$; learning such a coefficient costs about as much as ignoring it. When $n_{T,j}r_j^2>1$, the term is close to $2\log(r_j\sqrt{n_{T,j}})$, twice the logarithm of the ratio between the prior range $r_j$ and the accuracy $n_{T,j}^{-1/2}$ with which $n_{T,j}$ observations determine $\theta_j$.

The factor $\frac12$ in \eqref{eq:intro-result} is exact precisely when the second regime dominates. Let
\[
 d_T(r)=\sum_{j\le T}\min\{1,n_{T,j}r_j^2\}.
\]
\Cref{thm:leading} shows that $\cR_T(r)\sim\frac12\Gamma_T(r)$ if and only if $d_T(r)=o(\Gamma_T(r))$, also for envelopes that change with $T$. The necessity of the condition follows from the sharper upper bound $\cR_T(r)\le\frac12\Gamma_T(r)-\delta\,d_T(r)$ of \Cref{lem:gap}, with an absolute constant $\delta>0$. \Cref{cor:leading-examples} gives $\cR_T(r)=\frac k2\log T+O(1)$ for a nonzero envelope with $k$ positive entries, the parametric constant, and $\cR_T(r)\sim\log^2(A^2T)/(8\alpha)$ for $r_j=Ae^{-\alpha j}$. For polynomial envelopes $d_T(r)/\Gamma_T(r)$ stays bounded away from zero, so $\cR_T(r)/\Gamma_T(r)$ stays below $\frac12$.

\paragraph{The lower bound.}
For any prior supported in the class, the worst-case regret of every predictor is at least the mutual information between the coefficients and the observations. We place the prior on the smaller envelope $r/\max\{1,B\}$, which keeps every log-odds in $[-1,1]$, and bound the information in two steps. In the first step we condition on the inputs. Conditionally on the coefficients and the inputs, the marks are independent logistic observations, and the design matrix has the lagged input vector $(U_t,U_{t-1},\ldots,U_1,0,\ldots,0)$ as its $t$-th row. \Cref{lem:information} applies a sequential form of the van Trees inequality \citep{gilllevit1995} to the prediction of each mark probability and sums the resulting bounds. The resulting sum is at least a constant multiple of the logarithm of the determinant of the identity plus a scaled Gram matrix of the design. Because the prior contributes the identity term, this quantity is finite for every design, including designs with linearly dependent columns. In the second step, \Cref{lem:gram} bounds its expectation over the random inputs. The diagonal entries of the Gram matrix are the deterministic counts $n_{T,j}$, while its off-diagonal entries have mean zero and variance $\min\{n_{T,j},n_{T,k}\}$, although they are not independent. After dividing by the diagonal, the condition $\sum_jr_j\le B$ bounds the expected squared Frobenius norm of the off-diagonal remainder both by a constant and by a constant multiple of $\Gamma_T(r)$. The first bound controls the probability that the remainder is large, and the second controls its effect on the determinant when it is small. Together they give a lower bound of order $\Gamma_T(r)$ even when every $n_{T,j}r_j^2$ is small, as in the example \eqref{eq:weak-example}.

\paragraph{The upper bound and the predictor.}
The upper bound uses the Bayesian mixture with independent Gaussian priors $N(0,r_j^2)$, together with the variational bound for Bayesian mixtures \citep{kakadeng2005,zhang2006}. In that bound we compare the mixture with a Gaussian distribution whose $j$-th coordinate has variance $r_j^2/(1+n_{T,j}r_j^2/4)$, and each lag then costs at most $\frac12\log(1+n_{T,j}r_j^2)$. The prior need not be supported in the class, because the minimax regret allows every predictor. For the matching factor $\frac12$ in \Cref{thm:leading}, \Cref{prop:resolution} gives a second lower bound that compares prior entropy with the mean squared error of a maximum-likelihood estimator. For computation, we run the online Newton step \citep{hazanagarwalkale2007} on inputs scaled by $r_j$; its regret bound contains the log-determinant of the Gram matrix of the scaled features, whose diagonal entries are again $n_{T,j}r_j^2$.

\paragraph{Contributions.}
Here and below, subscripts on $O$, $\Theta$ and $\asymp$ list the quantities on which the hidden constants depend.
\begin{enumerate}
\item \Cref{thm:main} gives $c_*\Gamma_T(r)/\max\{1,B^2\}\le\cR_T(r)\le\frac12\Gamma_T(r)$ for every summable envelope and every horizon, with an absolute constant $c_*$; the factor $\frac12$ cannot be decreased.
\item \Cref{thm:leading} shows that $\cR_T(r)\sim\frac12\Gamma_T(r)$ if and only if $d_T(r)=o(\Gamma_T(r))$. \Cref{cor:leading-examples} evaluates this for finitely supported envelopes and for $r_j=Ae^{-\alpha j^\beta}$, $\beta>0$. \Cref{cor:poly} gives $\cR_T(r)=\Theta_{B,s}((A^2T)^{1/(2s)})$ for $r_j=Aj^{-s}$ with $s>1$.
\item \Cref{thm:ons} gives an online Newton predictor for a known horizon with regret $O_B(\Gamma_T(r))$, an explicit tolerance for approximate projections, and polynomial cost per round.
\end{enumerate}

\subsection{Related work}\label{sec:related}
For a regular $d$-dimensional parametric family the Bayesian redundancy is $(d/2)\log T+O(1)$ \citep{clarkebarron1990}, and \Cref{cor:leading-examples} recovers this constant for finitely supported envelopes. Online logistic regression is usually studied with arbitrary feature vectors and with regret measured against the best fixed parameter in hindsight. In that setting, Bayesian mixtures \citep{kakadeng2005}, the online Newton step \citep{hazanagarwalkale2007} and improper predictors \citep{foster2018} give upper bounds, \citet{shamir2020} proves upper and lower bounds when the dimension may grow with $T$, and \citet{jacquetshamirszpankowski2021} computes the minimax regret for categorical features. These results concern a different minimax problem, since their comparator is a fixed parameter chosen in hindsight rather than the true source. \Cref{tab:comparison} summarizes the objectives and the information each guarantee uses.

\begin{table}[t]
\centering
\caption{Objectives and guarantees of related analyses. The third row is \Cref{prop:window}, which applies the standard online Newton bound to a window of $h$ lags; it is not the strongest result of the cited work.}
\label{tab:comparison}
\small
\setlength{\tabcolsep}{3.3pt}
\begin{tabular}{@{}>{\raggedright\arraybackslash}p{0.2\textwidth}>{\raggedright\arraybackslash}p{0.25\textwidth}>{\raggedright\arraybackslash}p{0.21\textwidth}>{\raggedright\arraybackslash}p{0.28\textwidth}@{}}
\toprule
Analysis & Objective & What the bound uses & Guarantee \\
\midrule
Parametric redundancy \citep{clarkebarron1990}
 & Expected regret for a fixed regular family
 & The dimension $d$
 & $(d/2)\log T+O(1)$ \\
\addlinespace
Online logistic regression \citep{kakadeng2005,foster2018,shamir2020,jacquetshamirszpankowski2021}
 & Arbitrary features and labels; regret against the best fixed parameter in hindsight
 & Dimension and norm of the comparator
 & Comparator bounds for a different minimax problem \\
\addlinespace
Dimension-only window bound \citep[derived from][]{hazanagarwalkale2007}
 & Expected regret for the source \eqref{eq:intro-model}
 & Window length $h$ and the omitted tail
 & $\inf_h\{h\log T+T\sum_{j>h}r_j^2\}$ \\
\addlinespace
\textbf{This paper}
 & Expected regret for the source \eqref{eq:intro-model}
 & Each $r_j$ and its count $n_{T,j}$
 & $\Gamma_T(r)$ up to constants; $\cR_T\sim\frac12\Gamma_T$ if and only if $d_T=o(\Gamma_T)$ \\
\bottomrule
\end{tabular}
\end{table}

In nonparametric problems redundancy is governed by metric entropy \citep{yangbarron1999}, and smooth truncated coverings give sharp log-loss bounds for large classes of experts \citep{wu2022}; our upper bound applies the same localization idea one coordinate at a time. \citet{hanjiangwu2024} study infinite-memory processes through the loss of a single prediction made after a training sequence, whereas our regret adds the excess losses of all predictions along the sequence. \Cref{app:related} contains further comparisons.

\section{Model and Main Result}\label{sec:setup}
All logarithms are natural. We write $I_m$ for the $m\times m$ identity matrix, $\KL(P\|Q)$ for relative entropy, $x_+=\max\{x,0\}$ for real $x$, and $A\preceq C$ when $C-A$ is positive semidefinite. The norms $\norm{\cdot}_2$, $\norm{\cdot}_F$ and $\norm{\cdot}_{\mathrm{op}}$ are the Euclidean, Frobenius and operator norms, and $\diag(a_1,\ldots,a_m)$ is the diagonal matrix with diagonal entries $a_1,\ldots,a_m$.

\paragraph{The source.}
The inputs $U_1,U_2,\ldots$ are independent uniform signs. The first mark $Y_1$ is a uniform bit independent of everything else; it carries no information about $\theta$ and serves only to define $X_t=(U_t,Y_t)$ at $t=1$. For $t\ge1$, the mark $Y_{t+1}$ follows \eqref{eq:intro-model}. Equivalently, conditionally on $\theta$ and on the whole input sequence, the marks $Y_2,Y_3,\ldots$ are independent and $Y_{t+1}$ depends on the inputs only through $U_{1:t}$; in particular, $U_{t+1}$ is independent of $Y_{t+1}$ given $X_{1:t}$. The parameter set is
\begin{equation}
 \Theta_r=\{\theta:\abs{\theta_j}\le r_j\text{ for every }j\},
 \qquad r_j\ge0,\quad \sum_{j\ge1}r_j\le B. \label{eq:class}
\end{equation}
We write $P_{\theta,T}$ for the joint law of $X_{1:T+1}$. By \eqref{eq:intro-model}, the marks $Y_2,\ldots,Y_{T+1}$ use lags $j\le T$ only, so $P_{\theta,T}$ depends on $\theta$ only through $(\theta_1,\ldots,\theta_T)$. As in \Cref{sec:intro}, $n_{T,j}=(T-j+1)_+$ for $j\ge1$; it is zero for $j>T$. We also set $U_k=0$ for $k\le0$, so that the input history at time $1$ is empty.

\paragraph{Regret.}
A predictor $Q$ is a sequence of conditional probability mass functions $q_t(\cdot\mid x_{1:t})$ on $\{-1,1\}\times\{0,1\}$, $t=1,\ldots,T$; the function $q_t$ is the prediction of $X_{t+1}$ after $X_{1:t}$ has been observed. The first observation $X_1$ has the same law under $Q$ as under the source, and $Q(x_{1:T+1})$ denotes the resulting joint probability. Let $p_{\theta,t}(x_{t+1}\mid x_{1:t})=P_{\theta,T}(X_{t+1}=x_{t+1}\mid X_{1:t}=x_{1:t})$. The regret of $Q$ at $\theta$ and the minimax regret over $\Theta_r$ are
\begin{equation}
 \Reg_T(Q;\theta)=\E_\theta\sum_{t=1}^{T}
 \log\frac{p_{\theta,t}(X_{t+1}\mid X_{1:t})}{q_t(X_{t+1}\mid X_{1:t})},
 \qquad \cR_T(r)=\inf_Q\sup_{\theta\in\Theta_r}\Reg_T(Q;\theta). \label{eq:regret}
\end{equation}
The law $P_{v,T}$ is defined by \eqref{eq:intro-model} for every real coefficient sequence $v$, also outside $\Theta_r$, and $p_{v,t}$ denotes its conditional probabilities. For a prior $\Pi$ on $\Theta_r$, let $P_\Pi=\int P_{\theta,T}\,\Pi(d\theta)$ and let $I_\Pi(\theta;X_{1:T+1})$ be the mutual information under the joint law in which $\theta\sim\Pi$ and $X_{1:T+1}\sim P_{\theta,T}$ given $\theta$. The next lemma collects classical identities \citep{davisson1973,haussler1997} in the form used below; \Cref{app:base-radius} proves them for the present source.

\begin{Lem}[Regret and mutual information]\phantomsection\label{lem:identities}
\leavevmode
\begin{enumerate}[label=(\alph*)]
\item For every predictor $Q$ and every $\theta$, $\Reg_T(Q;\theta)=\KL(P_{\theta,T}\|Q)$. A predictor that uses an independent random seed has worst-case expected regret at least $\cR_T(r)$.
\item For every predictor $Q$, every $\theta$ and every real sequence $v$,
\begin{equation}
 \Reg_T(Q;\theta)=\KL(P_{\theta,T}\|P_{v,T})+\E_\theta\sum_{t=1}^{T}\log\frac{p_{v,t}(X_{t+1}\mid X_{1:t})}{q_t(X_{t+1}\mid X_{1:t})}. \label{eq:regret-split}
\end{equation}
\item For every prior $\Pi$ on $\Theta_r$ and every predictor $Q$,
\begin{equation}
 \E_{\theta\sim\Pi}\KL(P_{\theta,T}\|Q)
 = I_\Pi(\theta;X_{1:T+1})+\KL(P_\Pi\|Q)
 \ge I_\Pi(\theta;X_{1:T+1}), \label{eq:information-radius}
\end{equation}
and consequently
\begin{equation}
 \cR_T(r)\ge I_\Pi(\theta;X_{1:T+1}). \label{eq:radius-lower}
\end{equation}
\item For every prior $\Pi$ on $\Theta_r$,
\begin{equation}
 I_\Pi(\theta;X_{1:T+1})
 =I_\Pi(\theta;Y_{2:T+1}\mid U_{1:T})
 =\E_{U_{1:T}}\bigl[I_\Pi(\theta;Y_{2:T+1}\mid U_{1:T}=u)\bigr].
 \label{eq:conditional-information}
\end{equation}
For each fixed $u$, the marks $Y_2,\ldots,Y_{T+1}$ are conditionally independent given $\theta$, and the logit of $Y_{t+1}$ is $\sum_{j\le t}\theta_ju_{t+1-j}$.
\end{enumerate}
\end{Lem}

All lower bounds in this paper start from (c) and (d). If a predictor gives each new input its uniform law, the input factors cancel in the last sum of \eqref{eq:regret-split}, and that sum involves only the predictions of the marks; the predictors constructed for the upper bounds are of this kind.

\paragraph{The main theorem.}
Let $\psi(z)=\log(1+e^z)$, so that $\psi'=\sig$ and $\psi''=\sig(1-\sig)$, and define
\begin{equation}
 \kappa_B=\min_{\abs z\le B}\psi''(z)=\sig(B)(1-\sig(B))>0,
 \qquad s_{T,j}=n_{T,j}r_j^2,\qquad
 d_T(r)=\sum_{j\le T}\min\{1,s_{T,j}\}. \label{eq:scales}
\end{equation}
For $s\ge1$ we have $\log(1+s)\ge\log2$, and for $s\in[0,1]$ concavity of $s\mapsto\log(1+s)$ gives $\log(1+s)\ge s\log2$. Applying these two inequalities term by term,
\begin{equation}
 (\log2)\,d_T(r)\le\Gamma_T(r). \label{eq:d-gamma}
\end{equation}

\begin{Thm}[Minimax regret for every summable envelope]\phantomsection\label{thm:main}
For every $B\ge0$, every $r$ satisfying \eqref{eq:class}, and every integer $T\ge1$,
\begin{equation}
 \frac{c_*}{\max\{1,B^2\}}\,\Gamma_T(r)\le\cR_T(r)\le\frac12\Gamma_T(r),
 \qquad c_*=\frac{2\kappa_1^2}{\pi^2}\Bigl(1-\frac1{2\pi^2}\Bigr)\approx0.0074. \label{eq:main}
\end{equation}
The upper bound is attained by a single predictor that does not depend on $T$: the Bayesian mixture whose prior has independent coordinates, Gaussian $N(0,r_j^2)$ when $r_j>0$ and equal to zero when $r_j=0$.
\end{Thm}

The upper bound is proved in \Cref{sec:upper} and the lower bound in \Cref{sec:lower-proof}. Neither constant depends on which lags carry positive $r_j$ or on the sizes of the $r_j$; a coefficient at a late lag $j$ is charged through the number $n_{T,j}$ of rounds in which it appears. The factor $\frac12$ in the upper bound cannot be decreased: for a nonzero envelope with finitely many positive entries, $\cR_T(r)/\Gamma_T(r)\to\frac12$ as $T\to\infty$ (\Cref{cor:leading-examples}). Two special cases show the two regimes of $\log(1+x)$. If $s_{T,j}\le1$ for all $j$, then $\Gamma_T(r)$ lies between $(\log2)\sum_js_{T,j}$ and $\sum_js_{T,j}$, so $\cR_T(r)\asymp_B\sum_{j\le T}n_{T,j}r_j^2$. If $r$ has exactly $k$ positive entries, then $\Gamma_T(r)=k\log T+O_r(1)$ as $T\to\infty$, the order of the regret for a $k$-dimensional parametric model.

\paragraph{Relative entropy between two sources.}
For $S\subseteq\{1,\ldots,T\}$, define $\theta^S$ by $\theta^S_j=\theta_j$ for $j\in S$ and $\theta^S_j=0$ for $j\notin S$. The next lemma computes $\KL(P_{\theta,T}\|P_{v,T})$ up to constants; we use the upper inequality in the upper bounds of \Cref{sec:upper,sec:algorithm}. Its proof is in \Cref{app:pairwise}.

\begin{Lem}[Relative entropy between two sources]\phantomsection\label{lem:pairwise}
For $\theta,v\in\Theta_r$,
\begin{equation}
 \frac{\kappa_B}{2}\sum_{j\le T}n_{T,j}(\theta_j-v_j)^2
 \le\KL(P_{\theta,T}\|P_{v,T})
 \le\frac18\sum_{j\le T}n_{T,j}(\theta_j-v_j)^2. \label{eq:pairwise}
\end{equation}
The upper inequality holds for any two real coefficient vectors. In particular, $\KL(P_{\theta,T}\|P_{\theta^S,T})\le\frac18\sum_{j\notin S}n_{T,j}\theta_j^2$ for every $S$.
\end{Lem}

\section{Proof of the Main Theorem}\label{sec:bounds}
\Cref{sec:upper} proves the upper bound in \eqref{eq:main}. The lower bound has three parts: an information inequality for logistic observations with a fixed design (\Cref{sec:information}), a bound on the expected log-determinant of the lag design (\Cref{sec:gram}), and their combination (\Cref{sec:lower-proof}).

\subsection{Upper bound}\label{sec:upper}
For $r_j>0$ let $\Pi_j$ be the Gaussian distribution $N(0,r_j^2)$, and for $r_j=0$ let $\Pi_j$ be the point mass at zero. Let $\Pi_T=\bigotimes_{j\le T}\Pi_j$ and let $Q_{\Pi_T}=\int P_{v,T}\,\Pi_T(dv)$ be the Bayesian mixture at horizon $T$. The law $P_{v,T}$ is defined by \eqref{eq:intro-model} for every real vector $v$, so $Q_{\Pi_T}$ is a predictor in the sense of \Cref{sec:setup} even though $\Pi_T$ is not supported in $\Theta_r$: the minimax regret allows every predictor, not only mixtures of sources in the class. The upper bound rests on the variational bound for Bayesian mixtures used by \citet{kakadeng2005} and \citet{zhang2006}; \Cref{app:base-mixture} gives its proof.

\begin{Lem}[Variational bound for Bayesian mixtures]\phantomsection\label{lem:mixture}
For every $\theta$ and every probability measure $\rho$ on $\R^T$ that is absolutely continuous with respect to $\Pi_T$, has finite first moments, and satisfies $\KL(\rho\|\Pi_T)<\infty$,
\begin{equation}
 \KL(P_{\theta,T}\|Q_{\Pi_T})
 \le \int\KL(P_{\theta,T}\|P_{v,T})\,\rho(dv)+\KL(\rho\|\Pi_T).
 \label{eq:mixture-variation}
\end{equation}
\end{Lem}

We apply \Cref{lem:mixture} with a Gaussian measure $\rho$ that depends on $\theta$.

\begin{proof}[Proof of the upper bound in \eqref{eq:main}]
Fix $\theta\in\Theta_r$ and abbreviate $n_j=n_{T,j}$ and $s_j=s_{T,j}=n_jr_j^2$.
Let $\rho=\bigotimes_{j\le T}\rho_j$. If $r_j=0$, then $\theta_j=0$ and we let $\rho_j=\Pi_j$. If $r_j>0$, we let $\rho_j=N(\mu_j,w_j)$ with
\begin{equation}
 \mu_j=\frac{s_j}{s_j+4}\,\theta_j,\qquad w_j=\frac{r_j^2}{1+s_j/4}. \label{eq:gaussian-choice}
\end{equation}
These are the minimizers over $(\mu_j,w_j)$ of the quantity $F_j$ computed below.
Relative entropy is additive over product measures, and for $r_j>0$ the relative entropy between one-dimensional Gaussian laws is
\begin{equation}
 \KL(\rho_j\|\Pi_j)=\frac12\left[\frac{\mu_j^2+w_j}{r_j^2}-1+\log\frac{r_j^2}{w_j}\right]. \label{eq:gauss-kl}
\end{equation}
Coordinates with $r_j=0$ contribute zero.
By the upper inequality in \eqref{eq:pairwise}, which holds for arbitrary real coefficient vectors,
\begin{equation}
 \int\KL(P_{\theta,T}\|P_{v,T})\,\rho(dv)
 \le\frac18\sum_{j\le T}n_j\,\E_{\rho_j}(v_j-\theta_j)^2
 =\frac18\sum_{j\le T}n_j\bigl[(\theta_j-\mu_j)^2+w_j\bigr]. \label{eq:gauss-div}
\end{equation}
Let $r_j>0$ and let $F_j$ be the sum of the $j$-th terms in \eqref{eq:gauss-kl} and \eqref{eq:gauss-div}. We group the terms that contain $w_j$ and those that contain $\mu_j$. Since $n_j=s_j/r_j^2$, the terms in $w_j$ are
\[
 \frac{n_jw_j}{8}+\frac{w_j}{2r_j^2}=\frac{w_j(s_j+4)}{8r_j^2}=\frac12 ,
\]
and the remaining $w_j$-terms of \eqref{eq:gauss-kl} give $-\frac12+\frac12\log(r_j^2/w_j)=-\frac12+\frac12\log(1+s_j/4)$. With $\theta_j-\mu_j=4\theta_j/(s_j+4)$ and $s_j^2/r_j^2=n_js_j$, the terms in $\mu_j$ are
\[
 \frac{n_j}{8}\cdot\frac{16\theta_j^2}{(s_j+4)^2}+\frac{s_j^2\theta_j^2}{2r_j^2(s_j+4)^2}
 =\frac{n_j\theta_j^2(4+s_j)}{2(s_j+4)^2}=\frac{n_j\theta_j^2}{2(s_j+4)} .
\]
Adding the three parts and using $\theta_j^2\le r_j^2$,
\begin{equation}
 F_j=\frac12\log\Bigl(1+\frac{s_j}4\Bigr)+\frac{n_j\theta_j^2}{2(s_j+4)}
 \le\frac12\Bigl[\log\Bigl(1+\frac{s_j}4\Bigr)+\frac{s_j}{s_j+4}\Bigr]. \label{eq:gaussian-coordinate}
\end{equation}
To bound the right side of \eqref{eq:gaussian-coordinate}, let $z=s_j/4$ and $u=z/(1+z)\in[0,1)$, so that $s_j/(s_j+4)=u$. Convexity of the exponential function on $[0,1]$ gives $e^u\le(1-u)+ue=1+(e-1)u$. Multiplying by $1+z$,
\[
 (1+z)\,e^{z/(1+z)}\le1+z+(e-1)z=1+ez .
\]
Taking logarithms, $\log(1+s_j/4)+s_j/(s_j+4)\le\log(1+es_j/4)$. With \eqref{eq:gaussian-coordinate}, $F_j\le\frac12\log(1+es_j/4)$, and $F_j\le\frac12\log(1+s_j)$ because $e/4<1$.
The measure $\rho$ has finite first moments and $\KL(\rho\|\Pi_T)<\infty$, so \Cref{lem:mixture} applies. Summing the bound $F_j\le\frac12\log(1+es_j/4)$ over $j$ gives $\KL(P_{\theta,T}\|Q_{\Pi_T})\le\frac12\sum_{j\le T}\log(1+es_j/4)$. Since $\Reg_T(Q;\theta)=\KL(P_{\theta,T}\|Q)$ and $\theta\in\Theta_r$ was arbitrary,
\begin{equation}
 \cR_T(r)\le\sup_\theta\Reg_T(Q_{\Pi_T};\theta)\le\frac12\sum_{j\le T}\log\Bigl(1+\frac e4s_{T,j}\Bigr)\le\frac12\Gamma_T(r). \label{eq:upper-sharp}
\end{equation}
It remains to show that a single predictor satisfies \eqref{eq:upper-sharp} at every horizon. Let $\Pi=\bigotimes_{j\ge1}\Pi_j$. For $T'\ge T$ and every real sequence $v$, the law of $X_{1:T+1}$ under $P_{v,T'}$ is $P_{v,T}$, which depends only on $v_1,\ldots,v_T$. Integrating over $\Pi$ therefore gives $Q_{\Pi_T}$ as the law of $X_{1:T+1}$ under the mixture $\int P_{v,T'}\,\Pi(dv)$. Consequently the sequential predictor that, after $X_{1:t}$, averages $p_{v,t}(\cdot\mid X_{1:t})$ over the posterior of $v$ under $\Pi$ has joint law $Q_{\Pi_T}$ on $X_{1:T+1}$ for every $T$, and \eqref{eq:upper-sharp} applies to it at every horizon. It uses only the observed history, and its prediction of the next input is uniform because each $p_{v,t}$ gives the input its uniform law.
\end{proof}

\subsection{An information inequality for logistic observations}\label{sec:information}
This subsection works with a general design. Let $m\ge1$, let $Z\in\R^{T\times m}$ be a random matrix with rows $z_1^\top,\ldots,z_T^\top$, and let $b\in\R^T$ be a random vector, with $(Z,b)$ independent of the parameter $\vartheta\in\R^m$. Given $(Z,b,\vartheta)$, the labels $W_1,\ldots,W_T\in\{0,1\}$ are independent and $W_t$ equals one with probability $p_t(\vartheta)=\sig(b_t+z_t^\top\vartheta)$. Fix radii $r_1,\ldots,r_m>0$, let $R=\diag(r_1,\ldots,r_m)$, and assume that almost surely
\begin{equation}
 \abs{b_t+z_t^\top v}\le B\qquad\text{for all }t\le T\text{ and all }v\in\textstyle\prod_{j\le m}[-r_j,r_j]. \label{eq:logit-bound}
\end{equation}
Fix $t\in\{1,\ldots,T\}$ and a value of $(Z,b)$, let $K=\prod_{j\le m}[-r_j,r_j]$, and let $f$ be a probability density on $K$. The pair $(\vartheta,W_{1:t-1})$ then has the joint density
\[
 \varphi(\vartheta,w)=f(\vartheta)\prod_{s<t}p_s(\vartheta)^{w_s}(1-p_s(\vartheta))^{1-w_s}
\]
with respect to Lebesgue measure on $K$ times counting measure on $\{0,1\}^{t-1}$. Its score and Fisher information matrix are
\begin{equation}
 \zeta=\nabla_\vartheta\log\varphi(\vartheta,W_{1:t-1})=\nabla\log f(\vartheta)+\sum_{s<t}(W_s-p_s(\vartheta))z_s,
 \qquad J_{t-1}=\E[\zeta\zeta^\top], \label{eq:joint-score}
\end{equation}
where the formula for $\zeta$ uses $\nabla\log p_s=(1-p_s)z_s$ and $\nabla\log(1-p_s)=-p_sz_s$. We take $f=\prod_{j\le m}f_j$ with
\begin{equation}
 f_j(v)=r_j^{-1}\cos^2\left(\frac{\pi v}{2r_j}\right)\ind\{\abs v\le r_j\}.
 \label{eq:cosine-prior}
\end{equation}
This density and its derivative vanish on the boundary of $K$, which removes the boundary terms in the integration by parts behind the van Trees inequality \citep{gilllevit1995}. The next lemma states that inequality in the form we need; its proof is in \Cref{app:information}.

\begin{Lem}[van Trees inequality]\phantomsection\label{lem:van-trees}
Fix $t$ and $(Z,b)$ satisfying \eqref{eq:logit-bound}, and let $f$ be given by \eqref{eq:cosine-prior}.
\begin{enumerate}[label=(\alph*)]
\item For every function $\widehat g$ of $W_{1:t-1}$ and every continuously differentiable $g:K\to\R$,
\[
 \E(\widehat g(W_{1:t-1})-g(\vartheta))^2\ge d^\top J_{t-1}^{-1}d,\qquad d=\E\nabla g(\vartheta).
\]
\item $J_{t-1}=\pi^2R^{-2}+\sum_{s<t}\E[\psi''(b_s+z_s^\top\vartheta)]\,z_sz_s^\top$.
\end{enumerate}
\end{Lem}

\begin{Lem}[Information in logistic observations]\phantomsection\label{lem:information}
Let $\vartheta$ have the density $f=\prod_{j\le m}f_j$ from \eqref{eq:cosine-prior}. Then
\begin{equation}
 I(\vartheta;W_{1:T}\mid Z,b)
 \ge8\kappa_B^2\,\E\log\det\left(I_m+\frac{1}{4\pi^2}RZ^\top ZR\right).
 \label{eq:design-information}
\end{equation}
\end{Lem}

\begin{proof}
Fix a value of $(Z,b)$ that satisfies \eqref{eq:logit-bound}. Because $(Z,b)$ is independent of $\vartheta$, the conditional density of $\vartheta$ given $(Z,b)$ is still $f=\prod_jf_j$. Until the final averaging step, all probabilities and expectations are conditional on this fixed $(Z,b)$, and we prove
\begin{equation}
 I(\vartheta;W_{1:T})\ge8\kappa_B^2\log\det\left(I_m+\frac{1}{4\pi^2}RZ^\top ZR\right). \label{eq:fixed-design}
\end{equation}
For $t=1,\ldots,T$ define
\[
 H_{t-1}=\pi^2R^{-2}+\frac14\sum_{s<t}z_sz_s^\top ,
\]
and let $J_{t-1}$ be the Fisher information matrix in \eqref{eq:joint-score}. By \Cref{lem:van-trees}(b),
\[
 J_{t-1}=\pi^2R^{-2}+\sum_{s<t}\E[\psi''(b_s+z_s^\top\vartheta)]\,z_sz_s^\top .
\]
Since $\psi''\le1/4$ everywhere,
\begin{equation}
 J_{t-1}\preceq H_{t-1}. \label{eq:joint-fisher}
\end{equation}
Let $\widehat p_t=\E[p_t(\vartheta)\mid W_{1:t-1}]$ be the Bayes predictive probability that $W_t=1$. \Cref{lem:van-trees}(a) with $g=p_t$ and the estimator $\widehat p_t$ gives
\[
 \E(\widehat p_t-p_t(\vartheta))^2\ge d_t^\top J_{t-1}^{-1}d_t,
 \qquad d_t=\E\nabla p_t(\vartheta).
\]
The chain rule gives $\nabla p_t(\vartheta)=\psi''(b_t+z_t^\top\vartheta)\,z_t$, so $d_t=\bar\kappa_tz_t$ with $\bar\kappa_t=\E\psi''(b_t+z_t^\top\vartheta)$. By \eqref{eq:logit-bound} every logit lies in $[-B,B]$, hence $\bar\kappa_t\ge\kappa_B$. Both matrices in \eqref{eq:joint-fisher} are positive definite, and matrix inversion reverses the order $\preceq$ between positive definite matrices, so $J_{t-1}^{-1}\succeq H_{t-1}^{-1}$. Therefore
\begin{equation}
 \E(\widehat p_t-p_t(\vartheta))^2\ge\kappa_B^2\,z_t^\top H_{t-1}^{-1}z_t. \label{eq:vt-step}
\end{equation}
By the chain rule for mutual information,
\[
 I(\vartheta;W_{1:T})=\sum_{t=1}^{T}I(\vartheta;W_t\mid W_{1:t-1}).
\]
Given $\vartheta$ and $W_{1:t-1}$, the label $W_t$ is Bernoulli with parameter $p_t(\vartheta)$; given $W_{1:t-1}$ alone, it is Bernoulli with parameter $\widehat p_t$. Hence $I(\vartheta;W_t\mid W_{1:t-1})=\E\KL(\Ber(p_t(\vartheta))\|\Ber(\widehat p_t))$. Pinsker's inequality for Bernoulli distributions, $\KL(\Ber(p)\|\Ber(p'))\ge2(p-p')^2$ for $p\in[0,1]$ and $p'\in(0,1)$ \citep{coverthomas2006}, applies with $p'=\widehat p_t$ because $\widehat p_t\in(0,1)$ is an average of logistic probabilities. With \eqref{eq:vt-step} it gives
\begin{equation}
 I(\vartheta;W_{1:T})\ge2\sum_{t=1}^{T}\E(\widehat p_t-p_t(\vartheta))^2
 \ge2\kappa_B^2\sum_{t=1}^{T}z_t^\top H_{t-1}^{-1}z_t. \label{eq:information-telescope}
\end{equation}
Since $H_t=H_{t-1}+\frac14z_tz_t^\top$, the matrix determinant lemma gives
\[
 \det H_t=\det H_{t-1}\,\Bigl(1+\frac14z_t^\top H_{t-1}^{-1}z_t\Bigr).
\]
Using $\log(1+x)\le x$,
\[
 \log\frac{\det H_t}{\det H_{t-1}}\le\frac14z_t^\top H_{t-1}^{-1}z_t .
\]
Summing over $t$ and inserting the result in \eqref{eq:information-telescope} gives $I(\vartheta;W_{1:T})\ge8\kappa_B^2\log(\det H_T/\det H_0)$. Finally $H_0=\pi^2R^{-2}$ and $H_T=H_0+\frac14Z^\top Z$, so
\[
 \frac{\det H_T}{\det H_0}=\det\bigl(I_m+\tfrac14H_0^{-1/2}Z^\top ZH_0^{-1/2}\bigr)
 =\det\left(I_m+\frac{1}{4\pi^2}RZ^\top ZR\right),
\]
because $H_0^{-1/2}=R/\pi$. This proves \eqref{eq:fixed-design}.
The conditional mutual information $I(\vartheta;W_{1:T}\mid Z,b)$ is the average of the fixed-design quantity in \eqref{eq:fixed-design} over the law of $(Z,b)$. Averaging both sides of \eqref{eq:fixed-design} gives \eqref{eq:design-information}.
\end{proof}

The prior term $\pi^2R^{-2}$ in $H_{t-1}$ keeps every inverse in \eqref{eq:vt-step} finite, also for coordinates that have been observed only a few times. For this reason the bound retains coefficients that cannot be estimated accurately from the data.

\subsection{The expected log-determinant for the lag design}\label{sec:gram}
For the source \eqref{eq:intro-model}, the design is the $T\times T$ matrix
\begin{equation}
 Z_{t,j}=\ind\{t\ge j\}U_{t+1-j},\qquad 1\le t,j\le T,
 \qquad G=Z^\top Z. \label{eq:lag-design}
\end{equation}
Row $t$ of $Z$ is $(U_t,U_{t-1},\ldots,U_1,0,\ldots,0)$, so $Z$ is lower triangular with constant diagonals. Its Gram matrix $G$ is in general neither diagonal nor of this form, but its first two moments are explicit. Write $n_j=n_{T,j}$. Then
\begin{equation}
 G_{jj}=n_j,\qquad \E G_{jk}=0,\qquad
 \E G_{jk}^2=\min\{n_j,n_k\}\quad(j\ne k).
 \label{eq:gram-moments}
\end{equation}
To prove \eqref{eq:gram-moments}, let $j\le k$. Then $G_{jk}=\sum_{t=k}^{T}U_{t+1-j}U_{t+1-k}$, and the substitution $i=t+1-k$ gives
\[
 G_{jk}=\sum_{i=1}^{n_k}U_{i+k-j}U_i .
\]
For $j=k$ each summand equals one, so $G_{jj}=n_j$. For $j<k$ put $\ell=k-j\ge1$. Each summand $U_{i+\ell}U_i$ is a product of two independent signs, so it has mean zero and square one, and $\E G_{jk}=0$. For $i<i'$, the index $i$ is smaller than each of $i+\ell$, $i'$ and $i'+\ell$, so $U_i$ appears exactly once in the product $U_{i+\ell}U_iU_{i'+\ell}U_{i'}$. As $U_i$ is independent of the other factors and has mean zero, the product has mean zero. Distinct summands are therefore uncorrelated, and $\E G_{jk}^2=n_k=\min\{n_j,n_k\}$. The case $j>k$ follows by symmetry of $G$.

\begin{Lem}[Expected log-determinant of the lag design]\phantomsection\label{lem:gram}
Let $r_1,\ldots,r_T\ge0$ satisfy $\sum_{j\le T}r_j\le B$, let $R=\diag(r_1,\ldots,r_T)$, and let $G$ be as in \eqref{eq:lag-design}. For every $\lambda>0$,
\begin{equation}
 \E\log\det(I_T+\lambda RGR)
 \ge(1-2\lambda B^2)\sum_{j=1}^{T}\log(1+\lambda n_jr_j^2).
 \label{eq:gram-lower}
\end{equation}
\end{Lem}

\begin{proof}
Set $x_j=\lambda n_jr_j^2$, $y_j=x_j/(1+x_j)\in[0,1)$, $\Lambda=\sum_j\log(1+x_j)$ and $D=\diag(1+x_1,\ldots,1+x_T)$. By \eqref{eq:gram-moments}, the $j$-th diagonal entry of $I_T+\lambda RGR$ is $1+\lambda r_j^2n_j=1+x_j$. Define the symmetric matrix
\begin{equation}
 E_0=D^{-1/2}(I_T+\lambda RGR)D^{-1/2}-I_T. \label{eq:normalized-error}
\end{equation}
Its diagonal is zero, so $\tr E_0=0$, and for $j\ne k$
\[
 (E_0)_{jk}=\frac{\lambda r_jr_kG_{jk}}{\sqrt{(1+x_j)(1+x_k)}}.
\]
Since $I_T+\lambda RGR=D^{1/2}(I_T+E_0)D^{1/2}$ and $\log\det D=\Lambda$,
\begin{equation}
 \log\det(I_T+\lambda RGR)=\Lambda+\log\det(I_T+E_0). \label{eq:logdet-split}
\end{equation}
For $j\ne k$, \eqref{eq:gram-moments} gives
\[
 \E(E_0)_{jk}^2=\frac{\lambda^2r_j^2r_k^2\min\{n_j,n_k\}}{(1+x_j)(1+x_k)}
 =\frac{y_jy_k}{\max\{n_j,n_k\}},
\]
where the second equality uses $y_jy_k=\lambda^2n_jn_kr_j^2r_k^2/((1+x_j)(1+x_k))$ and $n_jn_k/\max\{n_j,n_k\}=\min\{n_j,n_k\}$. For $j<k$ we have $n_j>n_k$, so
\begin{equation}
 \omega\defeq\E\norm{E_0}_F^2
 =\sum_{j\ne k}\frac{y_jy_k}{\max\{n_j,n_k\}}
 =2\sum_{j<k}\frac{y_jy_k}{n_j}. \label{eq:q-identity}
\end{equation}
We bound $\omega$ in two ways. First, the inequality $2\sqrt x\le1+x$ gives $y_j\le\sqrt{x_j}/2=\sqrt{\lambda n_j}\,r_j/2$. For $j<k$, $n_j\ge n_k$ implies $n_j\ge\sqrt{n_jn_k}$, so
\[
 \frac{y_jy_k}{n_j}\le\frac{y_j}{\sqrt{n_j}}\cdot\frac{y_k}{\sqrt{n_k}}
 \le\frac{\lambda}{4}r_jr_k .
\]
Summing over $j<k$ and using $\sum_{j<k}r_jr_k\le\frac12(\sum_jr_j)^2\le B^2/2$,
\begin{equation}
 \omega\le\frac\lambda2\sum_{j<k}r_jr_k\le\frac{\lambda B^2}{4}.
 \label{eq:q-absolute}
\end{equation}
Second, from $y_j\le x_j$ we get $y_j/\max\{n_j,n_k\}\le x_j/n_j=\lambda r_j^2$. Also $u/(1+u)\le\log(1+u)$ for $u\ge0$, because both sides vanish at $u=0$ and the derivative $(1+u)^{-2}$ of the left side is at most the derivative $(1+u)^{-1}$ of the right side; hence $y_k\le\log(1+x_k)$. Applying these bounds in \eqref{eq:q-identity} and using $\sum_jr_j^2\le(\sum_jr_j)^2\le B^2$,
\begin{equation}
 \omega\le\lambda\Bigl(\sum_jr_j^2\Bigr)\sum_ky_k\le\lambda B^2\Lambda.
 \label{eq:q-relative}
\end{equation}
Let $\cE=\{\norm{E_0}_{\mathrm{op}}\le1/2\}$. Since $\norm{E_0}_{\mathrm{op}}\le\norm{E_0}_F$, Markov's inequality gives
\[
 \Pp(\cE^c)\le\Pp\bigl(\norm{E_0}_F^2>1/4\bigr)\le4\omega.
\]
On $\cE$, the eigenvalues $e_1,\ldots,e_T$ of $E_0$ lie in $[-1/2,1/2]$. The difference $\log(1+u)-(u-u^2)$ vanishes at $u=0$ and has derivative $u(1+2u)/(1+u)$, which is nonpositive on $[-1/2,0]$ and nonnegative on $[0,1/2]$; hence $\log(1+u)\ge u-u^2$ on $[-1/2,1/2]$. Therefore
\[
 \log\det(I_T+E_0)=\sum_j\log(1+e_j)\ge\sum_je_j-\sum_je_j^2=\tr E_0-\norm{E_0}_F^2=-\norm{E_0}_F^2 ,
\]
and by \eqref{eq:logdet-split}, $\log\det(I_T+\lambda RGR)\ge\Lambda-\norm{E_0}_F^2$ on $\cE$.
The matrix $\lambda RGR=\lambda(ZR)^\top(ZR)$ is positive semidefinite, so all eigenvalues of $I_T+\lambda RGR$ are at least one and $\log\det(I_T+\lambda RGR)\ge0$ on $\cE^c$.
By the bounds on $\cE$ and on $\cE^c$,
\begin{equation}
 \E\log\det(I_T+\lambda RGR)\ge\E\bigl[(\Lambda-\norm{E_0}_F^2)\ind_{\cE}\bigr]
 \ge\Lambda\,\Pp(\cE)-\E\norm{E_0}_F^2\ge(1-4\omega)\Lambda-\omega, \label{eq:gram-combine}
\end{equation}
where the last step uses $\Pp(\cE^c)\le4\omega$. By \eqref{eq:q-absolute}, $4\omega\Lambda\le\lambda B^2\Lambda$, and by \eqref{eq:q-relative}, $\omega\le\lambda B^2\Lambda$. Hence the right side is at least $(1-2\lambda B^2)\Lambda$, which is \eqref{eq:gram-lower}.
\end{proof}

The two bounds on $\omega$ have different roles in \eqref{eq:gram-combine}. The bound \eqref{eq:q-absolute} controls the probability of $\cE^c$, while \eqref{eq:q-relative} makes the loss $\omega$ proportional to $\Lambda$. With \eqref{eq:q-absolute} alone, \eqref{eq:gram-combine} would give $(1-\lambda B^2)\Lambda-\lambda B^2/4$, which is negative when $\Lambda$ is small, that is, when all lags are weak.

\subsection{Proof of the lower bound}\label{sec:lower-proof}
The lower bound uses a prior on a smaller class. Shrinking the envelope by the factor $\gamma=1/\max\{1,B\}$ keeps every logit in $[-1,1]$, so the curvature constant is $\kappa_1$ rather than $\kappa_B$, and the price is the factor $\gamma^2$ in \eqref{eq:main}.

\begin{proof}[Proof of the lower bound in \eqref{eq:main}]
Let $\gamma=1/\max\{1,B\}$ and $\widetilde r_j=\gamma r_j$, so that $\sum_j\widetilde r_j\le1$ and, since $\gamma\le1$, $\Theta_{\widetilde r}\subseteq\Theta_r$. Let $J_+=\{j\le T:r_j>0\}$. If $J_+$ is empty, then $\Gamma_T(r)=0$ and the lower bound holds because regret is a relative entropy and hence nonnegative. Otherwise let $m=\abs{J_+}$, give $\vartheta=(\theta_j)_{j\in J_+}$ the density $\prod_{j\in J_+}f_j$ from \eqref{eq:cosine-prior} with the radii $\widetilde r_j$, and set $\theta_j=0$ for all $j\notin J_+$. This prior is supported in $\Theta_{\widetilde r}\subseteq\Theta_r$, so \eqref{eq:radius-lower} gives $\cR_T(r)\ge I(\theta;X_{1:T+1})$.
The inputs $U_{1:T+1}$ and the first mark $Y_1$ are independent of $\theta$, and given $(\theta,U_{1:T+1},Y_1)$ the law of $Y_{2:T+1}$ depends only on $(\theta,U_{1:T})$. \Cref{lem:identities}(d) therefore gives $I(\theta;X_{1:T+1})=I(\theta;Y_{2:T+1}\mid U_{1:T})$.
Conditionally on $(\theta,U_{1:T})$, the labels $W_t=Y_{t+1}$, $t=1,\ldots,T$, are independent, and by \eqref{eq:intro-model} $W_t=1$ with probability $\sig(\sum_{j\in J_+}Z_{t,j}\theta_j)$, where $Z$ is the matrix \eqref{eq:lag-design}. This is the setting of \Cref{lem:information} with $b=0$, the columns of $Z$ indexed by $J_+$, the radii $\widetilde r_j$, and logit bound $1$ in place of $B$: since $\abs{Z_{t,j}}\le1$, every logit satisfies $\abs{\sum_{j\in J_+}Z_{t,j}v_j}\le\sum_j\widetilde r_j\le1$. The design depends only on the inputs and is therefore independent of $\vartheta$. Columns with $r_j=0$ contribute zero rows and columns, so the determinant in \eqref{eq:design-information} equals $\det(I_T+\lambda_0\widetilde RG\widetilde R)$ with $\widetilde R=\diag(\widetilde r_1,\ldots,\widetilde r_T)$ and $\lambda_0=(4\pi^2)^{-1}$. \Cref{lem:information} therefore gives
\begin{equation}
 \cR_T(r)\ge8\kappa_1^2\,\E\log\det\bigl(I_T+\lambda_0\widetilde RG\widetilde R\bigr). \label{eq:lower-step}
\end{equation}
Apply \eqref{eq:gram-lower} to the radii $\widetilde r_j$, whose sum is at most $1$, with $\lambda=\lambda_0$. Since $\widetilde r_j^2n_j=\gamma^2s_{T,j}$,
\[
 \E\log\det(I_T+\lambda_0\widetilde RG\widetilde R)\ge(1-2\lambda_0)\sum_{j\le T}\log(1+\lambda_0\gamma^2s_{T,j}),
 \qquad 1-2\lambda_0=1-\frac1{2\pi^2}.
\]
For $\lambda\in(0,1]$ and $x\ge0$, concavity of the logarithm gives
\[
 \log(1+\lambda x)=\log\bigl((1-\lambda)\cdot1+\lambda(1+x)\bigr)\ge\lambda\log(1+x).
\]
We apply this with $\lambda=\lambda_0\gamma^2\le1$ to each term. Combining \eqref{eq:lower-step} with the last two displays,
\begin{equation}
 \cR_T(r)\ge8\kappa_1^2\Bigl(1-\frac1{2\pi^2}\Bigr)\frac{\gamma^2}{4\pi^2}\Gamma_T(r)
 =\frac{c_*}{\max\{1,B^2\}}\Gamma_T(r). \label{eq:lower-finish}
\end{equation}
\end{proof}

\section{Rates and Leading Constants}\label{sec:rates}
By \Cref{thm:main}, the order of the minimax regret for a given envelope is the order of $\Gamma_T(r)$. \Cref{sec:rate-order} evaluates $\Gamma_T(r)$ for polynomial envelopes and compares it with the dimension-only window bound. \Cref{sec:leading} shows that the factor $\frac12$ of the upper bound in \eqref{eq:main} is the exact leading constant if and only if $d_T(r)=o(\Gamma_T(r))$.

\subsection{Polynomial envelopes and the dimension-only window bound}\label{sec:rate-order}
Let $r_j=Aj^{-s}$ with $s>1$ and put $M_T=(A^2T)^{1/(2s)}$. Then $n_{T,j}r_j^2=(1-(j-1)/T)(M_T/j)^{2s}$ for $j\le T$, so the term of lag $j$ in $\Gamma_T(r)$ is at least $\log(3/2)$ when $j\le M_T\le T/2$ and is at most $(M_T/j)^{2s}$ for every $j$. Thus $M_T$ counts, up to constants, the lags that contribute at least a constant to $\Gamma_T(r)$.

\begin{Cor}[Polynomial envelopes]\phantomsection\label{cor:poly}
Let $r_j=Aj^{-s}$ with $A>0$, $s>1$ and $\sum_jr_j\le B$. If $2\le M_T\le T/2$, then
\begin{equation}
 c_BM_T\le\cR_T(r)\le C_sM_T \label{eq:poly-order}
\end{equation}
with positive constants $c_B$ and $C_s$ that do not depend on $A$ or $T$. For every $T\ge1$,
\begin{equation}
 \Gamma_T(r)\le\frac{\pi}{\sin(\pi/(2s))}M_T. \label{eq:poly-gamma-upper}
\end{equation}
For fixed $A$ and $s$, as $T\to\infty$,
\begin{equation}
 \Gamma_T(r)\sim\frac{\pi}{\sin(\pi/(2s))}M_T,
 \qquad \cR_T(r)=\Theta_{B,s}(M_T). \label{eq:poly-equivalent}
\end{equation}
\end{Cor}

The proof is in \Cref{app:poly}. The equivalence in \eqref{eq:poly-equivalent} is an exact statement about $\Gamma_T(r)$; for the regret, \Cref{thm:main} determines the order but not the constant. The order $M_T$ has a direct explanation. The lags $j\le M_T$ contribute at most $\sum_{j\le M_T}\log(1+(M_T/j)^{2s})=O_s(M_T)$, and the lags $j>M_T$ contribute at most $A^2T\sum_{j>M_T}j^{-2s}=O_s(M_T)$.

A conventional analysis fits the first $h$ coefficients, charges each of them about $\log T$ as in the standard $h$-dimensional online Newton bound \citep{hazanagarwalkale2007}, and adds the loss caused by setting the other coefficients to zero. The next proposition records the bound this gives and its rate for polynomial envelopes; its proof is in \Cref{app:window}.

\begin{Prop}[Dimension-only window bound]\phantomsection\label{prop:window}
Let $h\ge1$. The online Newton step of \citet{hazanagarwalkale2007}, run on the unscaled inputs $U_{t+1-j}$, $j\le h$, with coefficients restricted to $\prod_{j\le h}[-r_j,r_j]$, has regret at most
\begin{equation}
 C_Bh\log(1+T)+\frac18\sum_{j>h}n_{T,j}r_j^2 ,\qquad C_B=\frac{B^2}{2\log2}+\frac1{2\kappa_B},
 \label{eq:window-benchmark}
\end{equation}
for every $\theta\in\Theta_r$. For $r_j=Aj^{-s}$ with fixed $A>0$ and $s>1$, as $T\to\infty$,
\begin{equation}
 \inf_{1\le h\le T}\{h\log T+TA^2h^{1-2s}\}
 \asymp_{A,s}T^{1/(2s)}(\log T)^{1-1/(2s)},
 \label{eq:window-rate}
\end{equation}
and the minimum of \eqref{eq:window-benchmark} over $h$ has the same order, with constants depending on $A$, $s$ and $B$.
\end{Prop}

The second term of \eqref{eq:window-benchmark} is the bound of \Cref{lem:pairwise} on the relative entropy between the source and the source with $\theta_j=0$ for $j>h$. By \Cref{cor:poly}, the minimum of \eqref{eq:window-benchmark} exceeds the minimax regret by the factor $(\log T)^{1-1/(2s)}$. The excess comes from the charge $\log(1+T)$ for every retained lag, not from the use of a finite window: the online Newton predictor of \Cref{sec:algorithm}, run on the first $\lceil M_T\rceil$ lags with features scaled by $r_j$, has regret $O_{B,s}(M_T)$ (see the discussion after \Cref{thm:ons}).

The lower bound of \Cref{thm:main} also applies when no coefficient is resolved relative to its allowed range, that is, when $\max_jn_{T,j}r_j^2\to0$. Fix $B>0$ and, for each $T$, let $r_j=B/m$ for $j\le m$ and $r_j=0$ for $j>m$, where $m=m_T$ satisfies $\sqrt T/m\to0$ and $m/T\to0$. Then
\begin{equation}
 \max_j n_{T,j}r_j^2=\frac{B^2T}{m^2}\longrightarrow0,
 \qquad \Gamma_T(r)\sim\frac{B^2T}{m},
 \qquad \cR_T(r)\asymp_B\frac{B^2T}{m}\longrightarrow\infty. \label{eq:weak-example}
\end{equation}
To verify \eqref{eq:weak-example}, note first that
\[
 \sum_{j=1}^{m}n_{T,j}r_j^2=\frac{B^2}{m^2}\Bigl(mT-\frac{m(m-1)}2\Bigr)=\frac{B^2T}{m}\Bigl(1-\frac{m-1}{2T}\Bigr)\sim\frac{B^2T}{m},
\]
because $m/T\to0$. The largest term is $n_{T,1}r_1^2=B^2T/m^2$, which tends to zero because $\sqrt T/m\to0$. For $0\le x\le\eps$ we have $(1-\eps/2)x\le x-x^2/2\le\log(1+x)\le x$, so $\Gamma_T(r)=(1+o(1))\sum_jn_{T,j}r_j^2\sim B^2T/m$. \Cref{thm:main} then gives $\cR_T(r)\asymp_B B^2T/m$, which tends to infinity because $m/T\to0$. In this example a lower bound that keeps only the lags with $n_{T,j}r_j^2\ge1$ would be zero. \Cref{thm:main} gives the correct order because \Cref{lem:gram} keeps every lag in the determinant, and the bound \eqref{eq:q-relative} makes its error proportional to $\Gamma_T(r)$.

\subsection{The leading constant}\label{sec:leading}
The constant $c_*/\max\{1,B^2\}$ of the lower bound in \eqref{eq:main} is far below $\frac12$. This subsection shows that $\frac12$ is the leading constant exactly when $d_T(r)=o(\Gamma_T(r))$. One direction comes from a sharper form of the upper bound.

\begin{Lem}[Gap below one half]\phantomsection\label{lem:gap}
For every envelope satisfying \eqref{eq:class} and every $T\ge1$,
\begin{equation}
 \cR_T(r)\le\frac12\Gamma_T(r)-\delta\,d_T(r),\qquad \delta=\frac12\log\frac{8}{4+e}\approx0.0873 . \label{eq:gap}
\end{equation}
\end{Lem}

\begin{proof}
Let $c_0=e/4<1$ and $g(s)=\log(1+s)-\log(1+c_0s)$ for $s\ge0$. By \eqref{eq:upper-sharp}, $\cR_T(r)\le\frac12\Gamma_T(r)-\frac12\sum_{j\le T}g(s_{T,j})$, so it suffices to show $g(s)\ge g(1)\min\{1,s\}$, since $\frac12g(1)=\frac12\log(2/(1+c_0))=\delta$. We have $g(0)=0$,
\[
 g'(s)=\frac{1-c_0}{(1+s)(1+c_0s)}>0,\qquad
 g''(s)=-\frac{(1-c_0)(1+c_0+2c_0s)}{(1+s)^2(1+c_0s)^2}<0 .
\]
For $s\in[0,1]$, concavity and $g(0)=0$ give $g(s)\ge s\,g(1)$. For $s\ge1$, monotonicity gives $g(s)\ge g(1)$.
\end{proof}

For the other direction we use a second lower bound. It compares the entropy of a uniform prior with an upper bound on the posterior entropy, obtained from the mean squared error of a constrained maximum-likelihood estimator; the Gaussian maximum-entropy inequality converts the error bound into an entropy bound. The method is classical for parametric redundancy \citep{clarkebarron1990,yangbarron1999}, and it uses no normal approximation of the posterior.

Let $S\subseteq\{j\le T:r_j>0\}$, let $m=\abs S$, and define
\begin{equation}
 v_S=\sum_{\substack{j,k\in S\\j\ne k}}\frac{1}{\max\{n_{T,j},n_{T,k}\}},
 \qquad \eta_B=\frac12\log\left(\frac{8\pi e}{\kappa_{B/2}^2}\right).
 \label{eq:subset-variance}
\end{equation}

\begin{Prop}[Lower bound with leading factor one half]\phantomsection\label{prop:resolution}
For every nonempty $S$ as above,
\begin{equation}
 \cR_T(r)\ge (1-4v_S)_+
 \left[\frac12\sum_{j\in S}\log(n_{T,j}r_j^2)-\eta_Bm\right]_+.
 \label{eq:resolution-lower}
\end{equation}
\end{Prop}

The proof is in \Cref{app:resolution}. The quantity $v_S$ is the expected squared Frobenius distance between the identity and the Gram matrix of the lags in $S$, after each column has been divided by $\sqrt{n_{T,j}}$. It is small when the signal levels $s_{T,j}=n_{T,j}r_j^2$ of the lags in $S$ are large: if $s_{T,j}\ge\tau$ for all $j\in S$, then $1/\max\{n_{T,j},n_{T,k}\}\le1/\sqrt{n_{T,j}n_{T,k}}\le r_jr_k/\tau$, and therefore $v_S\le B^2/\tau$. Applying \Cref{prop:resolution} to the lags with $s_{T,j}\ge\tau$ gives the first part of the next theorem; \Cref{lem:gap} gives the converse in the second part.

\begin{Thm}[The leading constant]\phantomsection\label{thm:leading}
For every envelope satisfying \eqref{eq:class}, every $T\ge1$ and every $\tau\ge2$,
\begin{equation}
 \cR_T(r)\ge\Bigl(1-\frac{4B^2}{\tau}\Bigr)_+\Bigl[\frac12\Gamma_T(r)-\Bigl(\eta_B+\frac12\log(1+\tau)\Bigr)d_T(r)\Bigr]_+ .
 \label{eq:threshold-lower}
\end{equation}
Moreover, let $r^{(T)}$, $T\ge1$, be envelopes, possibly depending on $T$, with $\sum_jr_j^{(T)}\le B$ for every $T$ and $\Gamma_T(r^{(T)})>0$ for all large $T$. Then
\begin{equation}
 \cR_T(r^{(T)})\sim\frac12\Gamma_T(r^{(T)})
 \quad\Longleftrightarrow\quad
 d_T(r^{(T)})=o\bigl(\Gamma_T(r^{(T)})\bigr). \label{eq:leading-general}
\end{equation}
\end{Thm}

The proof is in \Cref{app:leading}. The condition $d_T=o(\Gamma_T)$ says that lags with bounded signal level contribute a negligible part of $\Gamma_T$. Precisely, let $s_{T,j}$ denote the signal levels of $r^{(T)}$, and for $\tau\ge1$ let $\Gamma_T^{\le\tau}=\sum_{j\le T:\,s_{T,j}\le\tau}\log(1+s_{T,j})$. Then $d_T=o(\Gamma_T)$ if and only if $\Gamma_T^{\le\tau}=o(\Gamma_T)$ for every fixed $\tau\ge1$. For the forward implication, let $C_\tau=\max\{1,\log(1+\tau)\}$; then $\log(1+s)\le s\le C_\tau s$ for $s<1$ and $\log(1+s)\le C_\tau$ for $1\le s\le \tau$, so $\Gamma_T^{\le\tau}\le C_\tau d_T$. For the reverse implication, split $d_T$ at $\tau$. The lags with $s_{T,j}\le \tau$ contribute at most $\Gamma_T^{\le\tau}/\log2$, because $\log(1+s)\ge(\log2)\min\{1,s\}$ (see \eqref{eq:d-gamma}). Each lag with $s_{T,j}>\tau$ contributes $1$ to $d_T$ and more than $\log(1+\tau)$ to $\Gamma_T$, so there are fewer than $\Gamma_T/\log(1+\tau)$ of them. Hence $\limsup_Td_T/\Gamma_T\le1/\log(1+\tau)$ for every $\tau$, and letting $\tau\to\infty$ gives $d_T=o(\Gamma_T)$.

The condition holds for the following two families, for which the leading term can be computed.

\begin{Cor}[Finite support and stretched-exponential envelopes]\phantomsection\label{cor:leading-examples}
For a fixed envelope with exactly $k\ge1$ positive entries,
\begin{equation}
 \cR_T(r)=\frac{k}{2}\log T+O_{B,r}(1). \label{eq:finite-leading}
\end{equation}
For $r_j=Ae^{-\alpha j^\beta}$ with fixed $A,\alpha,\beta>0$ and $\sum_jr_j\le B$, let $L_T=\log(A^2T)$. Then
\begin{equation}
 \cR_T(r)\sim
 \frac{\beta}{2(\beta+1)(2\alpha)^{1/\beta}}
 L_T^{1+1/\beta}. \label{eq:stretched-leading}
\end{equation}
In particular, for $\beta=1$,
\begin{equation}
 \cR_T(r)\sim\frac{\log^2(A^2T)}{8\alpha}. \label{eq:exponential-leading}
\end{equation}
\end{Cor}

The proof is in \Cref{app:leading-examples}. For the stretched-exponential envelope, put $\mu_T=(L_T/(2\alpha))^{1/\beta}$, the lag at which $TA^2e^{-2\alpha j^\beta}=1$. For $j\le\mu_T$,
\[
 \log(n_{T,j}r_j^2)=L_T-2\alpha j^\beta+\log\Bigl(1-\frac{j-1}{T}\Bigr),
\]
where $L_T-2\alpha j^\beta\ge0$ and the last terms sum to $O(\mu_T^2/T)=o(1)$ over $j\le\mu_T$. Moreover
\begin{equation}
 \sum_{j\le\mu_T}(L_T-2\alpha j^\beta)
 \sim\frac{\beta}{\beta+1}L_T\mu_T. \label{eq:leading-area}
\end{equation}
This sum is asymptotic to $\Gamma_T(r)$, whereas $d_T(r)=O(\mu_T)$, which is $o(L_T\mu_T)$ because $L_T\to\infty$. \Cref{thm:leading} then gives one half of \eqref{eq:leading-area}, and substituting $\mu_T$ gives \eqref{eq:stretched-leading}.

For polynomial envelopes the factor $\frac12$ is not attained. For $r_j=Aj^{-s}$ with $2\le M_T\le T/2$, every $j\le M_T$ has $n_{T,j}\ge T-M_T+1\ge T/2$ and therefore $s_{T,j}\ge\frac12(M_T/j)^{2s}\ge\frac12$, so $d_T(r)\ge\frac12\lfloor M_T\rfloor\ge M_T/4$, while $\Gamma_T(r)\le\pi M_T/\sin(\pi/(2s))$ by \eqref{eq:poly-gamma-upper}. With \Cref{lem:gap},
\[
 \frac{\cR_T(r)}{\Gamma_T(r)}\le\frac12-\frac{\delta\sin(\pi/(2s))}{4\pi}.
\]
The exact leading constant of $\cR_T(r)$ for polynomial envelopes remains open.

\section{An Online Newton Predictor}\label{sec:algorithm}
The mixture of \Cref{thm:main} attains the upper bound with one predictor for all horizons, but its predictions are integrals over the posterior of all retained coefficients. This section gives a predictor for a known horizon $T$ whose arithmetic cost per round is polynomial in $T$, the number of retained lags, and $B$. It is the online Newton step (ONS) of \citet{hazanagarwalkale2007}, run on the inputs scaled by the envelope, and it keeps only the lags whose term in $\Gamma_T(r)$ is at least $\log2$.

Fix $T$ and let $S_T=\{j\le T:n_{T,j}r_j^2\ge1\}$. For a set $S\subseteq\{j\le T:r_j>0\}$ we write the retained coefficients as $\theta_j=r_ju_j$ with $u$ in the box $K=[-1,1]^{\abs S}$, and we use the feature vector and loss
\begin{equation}
 x_t=(r_jU_{t+1-j})_{j\in S},\qquad
 \ell_t(u)=\psi(x_t^\top u)-Y_{t+1}x_t^\top u. \label{eq:ons-features}
\end{equation}
Here $\ell_t(u)$ is the log loss of the prediction $\sig(x_t^\top u)$ for $Y_{t+1}$. Since $U_k=0$ for $k\le0$, the entry of $x_t$ for lag $j$ is zero when $t<j$, and $x_t$ uses only inputs observed by round $t$. For every $u\in K$, $\abs{x_t^\top u}\le\sum_jr_j\le B$.

The update requires a projection onto $K$ in the norm $\norm v_H=(v^\top Hv)^{1/2}$ of a positive definite matrix $H$, which in general is computed only approximately. For $a\in\R^{\abs S}$ and $\eps\ge0$, we call $v\in K$ an \emph{$\eps$-projection} of $a$ in the norm of $H$ if
\begin{equation}
 \sup_{u\in K}\,(v-u)^\top H(v-a)\le\eps . \label{eq:certificate}
\end{equation}
The left side equals $v^\top H(v-a)+\norm{H(v-a)}_1$, so the condition can be checked in $O(\abs S^2)$ operations. The exact projection satisfies \eqref{eq:certificate} with $\eps=0$, because the optimality condition of the projection states that $(u-v)^\top H(v-a)\ge0$ for every $u\in K$.

\begin{algorithm2e}[t]
\small
\caption{Online Newton step with envelope-scaled features}\label{alg:ons}
\KwIn{envelope $r$, logit bound $B$, horizon $T$, tolerances $\eps_1,\ldots,\eps_T\ge0$}
Set $S=S_T$ and $m=\abs S$\;
\If{$m=0$}{Predict $1/2$ for every mark and the uniform law for every input, and stop\;}
Set $u_1=0\in\R^m$ and $H_0=I_m$, and observe $X_1=(U_1,Y_1)$\;
\For{$t=1,\ldots,T$}{
 Form $x_t$ by \eqref{eq:ons-features}\;
 Predict $\widehat p_t=\sig(x_t^\top u_t)$ for $Y_{t+1}$ and the uniform law for $U_{t+1}$\;
 Observe $X_{t+1}$ and set $g_t=(\widehat p_t-Y_{t+1})x_t$\;
 $H_t=H_{t-1}+\kappa_B g_tg_t^\top$ and $a_t=u_t-H_t^{-1}g_t$\;
 Let $u_{t+1}$ be an $\eps_t$-projection of $a_t$ in the norm of $H_t$\;
}
\end{algorithm2e}

The analysis rests on a deterministic regret bound for the online Newton step, in which the features and labels are arbitrary fixed sequences. It is the bound of \citet{hazanagarwalkale2007} (see also \citealp{cesabianchilugosi2006}), with the log-determinant kept instead of being bounded by $m\log(1+T)$; \Cref{app:ons-pathwise} gives the proof.

\begin{Lem}[Pathwise bound for the online Newton step]\phantomsection\label{lem:ons-pathwise}
Let $K\subset\R^m$ be a box of the form $\prod_{i\le m}[-c_i,c_i]$, and let $x_1,\ldots,x_T\in\R^m$ and $y_1,\ldots,y_T\in\{0,1\}$ satisfy $\abs{x_t^\top u}\le B$ for all $t$ and all $u\in K$. Let $\ell_t(u)=\psi(x_t^\top u)-y_tx_t^\top u$ and $\kappa=\kappa_B$. Set $u_1=0$ and $H_0=I_m$, and for $t=1,\ldots,T$ let
\[
 g_t=(\sig(x_t^\top u_t)-y_t)x_t,\qquad H_t=H_{t-1}+\kappa g_tg_t^\top,\qquad a_t=u_t-H_t^{-1}g_t,
\]
and let $u_{t+1}\in K$ be an $\eps_t$-projection of $a_t$ in the norm of $H_t$, that is, a point satisfying \eqref{eq:certificate} with $v=u_{t+1}$, $H=H_t$ and $a=a_t$, for given tolerances $\eps_t\ge0$. Then, for every $u\in K$,
\begin{equation}
 \sum_{t=1}^{T}(\ell_t(u_t)-\ell_t(u))
 \le\frac12\norm u_2^2+\frac{1}{2\kappa}
 \log\det\left(I_m+\sum_{t=1}^{T}x_tx_t^\top\right)+\sum_{t=1}^{T}\eps_t.
 \label{eq:ons-determinant}
\end{equation}
\end{Lem}

The theorem is stated for an arbitrary retained set $S$ and then specialized to $S_T$. Write $\mathrm{ONS}_{r,S}$ for \Cref{alg:ons} run with the retained set $S$ in place of $S_T$.

\begin{Thm}[Regret of the online Newton predictor]\phantomsection\label{thm:ons}
For every $S\subseteq\{j\le T:r_j>0\}$, all tolerances $\eps_t\ge0$, and every $\theta\in\Theta_r$,
\begin{equation}
 \Reg_T(\mathrm{ONS}_{r,S};\theta)
 \le\frac{\abs S}{2}
 +\frac{1}{2\kappa_B}\sum_{j\in S}\log(1+n_{T,j}r_j^2)
 +\frac18\sum_{j\notin S}n_{T,j}\theta_j^2+\sum_{t=1}^T\eps_t .
 \label{eq:ons-oracle}
\end{equation}
For $S=S_T$ and $\eps_t\le1/T$,
\begin{equation}
 \Reg_T(\mathrm{ONS}_{r,S_T};\theta)\le C_B^{\mathrm{ONS}}\Gamma_T(r),\qquad C_B^{\mathrm{ONS}}=\frac1{2\kappa_B}+\frac3{2\log2}. \label{eq:ons-final}
\end{equation}
When $S_T$ is empty, the predictor computes no projections and incurs no tolerance term.
\end{Thm}

The proof is in \Cref{app:ons}. It applies \Cref{lem:ons-pathwise} with the scaled features. The diagonal entries of $\sum_tx_tx_t^\top$ are the deterministic numbers $n_{T,j}r_j^2$, and Hadamard's inequality bounds the log-determinant by $\sum_{j\in S}\log(1+n_{T,j}r_j^2)$. The first three terms in \eqref{eq:ons-oracle} are the cost of the initial point, the cost of learning the retained coefficients, and, by \Cref{lem:pairwise}, the loss caused by setting the other coefficients to zero; the last is the cost of approximate projections.

\paragraph{A scaled window for polynomial envelopes.}
Let $r_j=Aj^{-s}$ with $s>1$, assume $2\le M_T\le T/2$ as in \Cref{cor:poly}, and take $S=\{1,\ldots,h\}$ with $h=\lceil M_T\rceil$ and $\eps_t\le1/T$. Each term of \eqref{eq:ons-oracle} is $O_{B,s}(M_T)$. First, $h/2\le M_T$. Second, the sum over $j\le h$ is at most $\Gamma_T(r)/(2\kappa_B)=O_{B,s}(M_T)$ by \eqref{eq:poly-gamma-upper}. Third, since $\theta_j^2\le r_j^2$, $n_{T,j}\le T$, $TA^2=M_T^{2s}$ and $h\ge M_T$,
\[
 \frac18\sum_{j>h}n_{T,j}\theta_j^2\le\frac{TA^2}8\int_{h}^\infty x^{-2s}\,dx=\frac{TA^2h^{1-2s}}{8(2s-1)}\le\frac{M_T}{8(2s-1)} .
\]
Hence a predictor that uses only the first $\lceil M_T\rceil$ lags attains the minimax order $M_T$ once its features are scaled by $r_j$. The factor $(\log T)^{1-1/(2s)}$ in \eqref{eq:window-rate} is therefore caused by the dimension-only bound $C_Bh\log(1+T)$ and not by truncation.

\paragraph{Computation.}
We count operations in the real-arithmetic model, in which an evaluation of $\sig$ costs one operation, and we assume that the radii $r_1,\ldots,r_T$ are given; forming $S_T$ then takes $O(T)$ comparisons. The retained set satisfies $m\le\Gamma_T(r)/\log2$, because each $j\in S_T$ contributes at least $\log2$ to $\Gamma_T(r)$, and $m\le B\sqrt T$, because $s_{T,j}\ge1$ implies $r_j\ge n_{T,j}^{-1/2}\ge T^{-1/2}$ while $\sum_jr_j\le B$. Each round updates $H_t$ and $H_t^{-1}$ by rank-one formulas, with $O(m^2)$ operations and $O(m^2)$ storage. An approximate projection can be computed by projected gradient descent, as the next lemma shows; its proof is in \Cref{app:projection}.

\begin{Lem}[Cost of an approximate projection]\phantomsection\label{lem:projection-cost}
Let $m>0$, fix a round $t$, and let $\nu=1+\kappa_BTB^2$. Projected gradient descent on $v\mapsto\frac12\norm{v-a_t}_{H_t}^2$ over $[-1,1]^m$, with step size $1/\nu$ and started at $u_t$, returns a $(1/T)$-projection of $a_t$ in the norm of $H_t$ after at most $1+2\nu\log_+\bigl(4B\sqrt{m\nu}\,T\bigr)$ iterations, where $\log_+x=\max\{\log x,0\}$. Each iteration costs $O(m^2)$ operations.
\end{Lem}

Since $\kappa_B\le\frac14$, a round therefore costs $O(m^2(1+TB^2)\log(2+mTB))$ operations. To form $x_t$, the predictor stores the most recent $\max S$ inputs, so the total storage is $O(m^2+\max S)$. If $r$ is nonincreasing, then $S_T=\{1,\ldots,m\}$ and the stored inputs number $m$.

\section{Conclusion}\label{sec:conclusion}
For the logistic source driven by independent sign inputs, the minimax cumulative log-loss regret over an envelope class lies between $c_*\Gamma_T(r)/\max\{1,B^2\}$ and $\frac12\Gamma_T(r)$. Each lag enters through its coefficient bound $r_j$ and the number $n_{T,j}$ of rounds in which it is observed, and both bounds hold at every horizon under summability of $r$ alone. The factor $\frac12$ is the exact leading constant if and only if $d_T(r)=o(\Gamma_T(r))$, which covers finitely supported and stretched-exponential envelopes. For polynomial envelopes $\cR_T(r)/\Gamma_T(r)$ stays below $\frac12$, and the exact leading constant remains open. For a known horizon, the online Newton predictor of \Cref{thm:ons} attains the order $\Gamma_T(r)$ with polynomial computation per round. The lower bound uses the independence and symmetry of the inputs through the moment identities \eqref{eq:gram-moments}. 

\section*{AI Disclosure}
Generative AI tools assisted with revising the exposition. The author is responsible for all mathematical claims, proofs, and references.

\setlength{\bibsep}{3pt}
\bibliography{refs}
\clearpage
\appendix
\crefalias{section}{appendix}
\crefalias{subsection}{appendix}
\section{Proofs of Lemmas~\ref{lem:identities}--\ref{lem:van-trees} and~\ref{lem:ons-pathwise}}\label{app:base}

\subsection{Proof of Lemma~\ref{lem:identities} (regret and mutual information)}\label{app:base-radius}
For (a), let $P_1$ be the common law of $X_1$. For a trajectory $x=(x_1,\ldots,x_{T+1})$,
\[
 P_{\theta,T}(x)=P_1(x_1)\prod_{t=1}^{T}p_{\theta,t}(x_{t+1}\mid x_{1:t}),
 \qquad Q(x)=P_1(x_1)\prod_{t=1}^{T}q_t(x_{t+1}\mid x_{1:t}).
\]
The factor $P_1(x_1)$ cancels in the ratio, so
\[
 \log\frac{P_{\theta,T}(x)}{Q(x)}=\sum_{t=1}^T\log\frac{p_{\theta,t}(x_{t+1}\mid x_{1:t})}{q_t(x_{t+1}\mid x_{1:t})}.
\]
Taking the expectation under $P_{\theta,T}$ gives $\KL(P_{\theta,T}\|Q)=\Reg_T(Q;\theta)$; both sides equal $+\infty$ if $Q$ gives probability zero to a trajectory that has positive probability under $P_{\theta,T}$. If a predictor uses an independent random seed $\xi$ and $Q_\xi$ is its joint law for the seed value $\xi$, then $\overline Q=\E_\xi Q_\xi$ is again a predictor in the sense of \Cref{sec:setup}: it gives $X_1$ the law $P_1$, and its conditional probabilities are those of the mixture. Convexity of relative entropy in its second argument gives $\KL(P_{\theta,T}\|\overline Q)\le\E_\xi\KL(P_{\theta,T}\|Q_\xi)$ for every $\theta$, so the worst-case expected regret of the randomized predictor is at least that of $\overline Q$, and hence at least $\cR_T(r)$.
For (b), add and subtract $\log p_{v,t}(X_{t+1}\mid X_{1:t})$ inside the sum in \eqref{eq:regret}. The expected sum of $\log(p_{\theta,t}/p_{v,t})$ equals $\KL(P_{\theta,T}\|P_{v,T})$ by (a) applied with $Q=P_{v,T}$, which gives \eqref{eq:regret-split}.
For (c), let $\Pi$ be a prior on $\Theta_r$. For each $\theta$, inserting $P_\Pi(x)$ in the logarithm splits the relative entropy as
\[
 \KL(P_{\theta,T}\|Q)=\sum_xP_{\theta,T}(x)\log\frac{P_{\theta,T}(x)}{P_\Pi(x)}
 +\sum_xP_{\theta,T}(x)\log\frac{P_\Pi(x)}{Q(x)} .
\]
The $\Pi$-average of the first sum is $\int\KL(P_{\theta,T}\|P_\Pi)\,\Pi(d\theta)=I_\Pi(\theta;X_{1:T+1})$. In the second sum the logarithm does not depend on $\theta$, so its $\Pi$-average is $\sum_xP_\Pi(x)\log(P_\Pi(x)/Q(x))=\KL(P_\Pi\|Q)\ge0$. This proves \eqref{eq:information-radius}. If $\Pi$ is supported in $\Theta_r$, the supremum of $\Reg_T(Q;\theta)$ over $\Theta_r$ is at least its $\Pi$-average, so (a) and \eqref{eq:information-radius} give \eqref{eq:radius-lower}.
For (d), by the chain rule for mutual information,
\[
 I_\Pi(\theta;X_{1:T+1})=I_\Pi(\theta;U_{1:T+1},Y_1)+I_\Pi(\theta;Y_{2:T+1}\mid U_{1:T+1},Y_1).
\]
The first term is zero because $(U_{1:T+1},Y_1)$ is independent of $\theta$. In the second term, the pair $(U_{T+1},Y_1)$ is independent of $(\theta,U_{1:T},Y_{2:T+1})$, so removing it from the conditioning does not change the conditional mutual information. This gives the first equality in \eqref{eq:conditional-information}, and the second is the definition of conditional mutual information. For fixed $u$ and $\theta$, the marks are independent by the description of the source in \Cref{sec:setup}, with the logits stated in (d).

\subsection{Proof of Lemma~\ref{lem:pairwise} (relative entropy between two sources)}\label{app:pairwise}
For real $a,b$, we use $\log\sig(z)=z-\psi(z)$ and $\log(1-\sig(z))=-\psi(z)$ in the definition of relative entropy:
\begin{align*}
 \KL(\Ber(\sig(a))\|\Ber(\sig(b)))
 &=\sig(a)\bigl[(a-\psi(a))-(b-\psi(b))\bigr]+(1-\sig(a))\bigl[\psi(b)-\psi(a)\bigr]\\
 &=\psi(b)-\psi(a)-\psi'(a)(b-a).
\end{align*}
Taylor's theorem with integral remainder then gives
\begin{equation}
 \KL(\Ber(\sig(a))\|\Ber(\sig(b)))
 =(b-a)^2\int_0^1(1-v)\,\psi''(a+v(b-a))\,dv.
 \label{eq:logistic-bregman}
\end{equation}
Since $\int_0^1(1-v)\,dv=\frac12$ and $0<\psi''\le\frac14$ on $\R$, the right side of \eqref{eq:logistic-bregman} is at most $(b-a)^2/8$. If $a,b\in[-B,B]$, the whole segment between them lies in $[-B,B]$, where $\psi''\ge\kappa_B$, and the right side is at least $\kappa_B(b-a)^2/2$.
Under $P_{\theta,T}$ and $P_{v,T}$, the law of $X_1$ and the conditional law of each new input are the same, and $U_{t+1}$ is independent of $Y_{t+1}$ given $X_{1:t}$. With $a_t=\sum_{j\le t}\theta_jU_{t+1-j}$ and $b_t=\sum_{j\le t}v_jU_{t+1-j}$, the chain rule for relative entropy therefore gives
\begin{equation}
 \KL(P_{\theta,T}\|P_{v,T})=\sum_{t=1}^{T}\E_\theta\,\KL\bigl(\Ber(\sig(a_t))\,\|\,\Ber(\sig(b_t))\bigr). \label{eq:pairwise-chain}
\end{equation}
Under $P_{\theta,T}$ the inputs are independent uniform signs, so $\E[U_iU_k]=\ind\{i=k\}$ and $\E_\theta(b_t-a_t)^2=\sum_{j\le t}(\theta_j-v_j)^2$. Exchanging the order of summation,
\[
 \sum_{t=1}^{T}\sum_{j\le t}(\theta_j-v_j)^2=\sum_{j=1}^{T}(T-j+1)(\theta_j-v_j)^2=\sum_{j\le T}n_{T,j}(\theta_j-v_j)^2 .
\]
Inserting the upper bound $(b_t-a_t)^2/8$ from \eqref{eq:logistic-bregman} into \eqref{eq:pairwise-chain} gives the upper inequality in \eqref{eq:pairwise}, for arbitrary real coefficient vectors. If $\theta,v\in\Theta_r$, then $\abs{a_t},\abs{b_t}\le\sum_jr_j\le B$, and the lower bound $\kappa_B(b_t-a_t)^2/2$ gives the lower inequality. Finally, $\theta^S\in\Theta_r$, and $\theta_j-\theta^S_j$ is zero for $j\in S$ and equals $\theta_j$ otherwise, which gives the bound on $\KL(P_{\theta,T}\|P_{\theta^S,T})$.

\subsection{Proof of Lemma~\ref{lem:mixture} (variational bound for Bayesian mixtures)}\label{app:base-mixture}
The observation space $\{-1,1\}^{T+1}\times\{0,1\}^{T+1}$ is finite. For every real $v$ and every $x$, $P_{v,T}(x)>0$, and more precisely
\begin{equation}
 -\log P_{v,T}(x)\le(T+1)\log4+T\sum_{j\le T}\abs{v_j}, \label{eq:loglik-floor}
\end{equation}
because $X_1$ and each new input have probability $\frac12$ per coordinate, each logit is at most $\sum_{j\le T}\abs{v_j}$ in absolute value, and $\sig(z)\ge\frac12e^{-\abs z}$ for real $z$. Let $\rho$ be absolutely continuous with respect to $\Pi_T$, with finite first moments and $\KL(\rho\|\Pi_T)<\infty$, and let $f_\rho=d\rho/d\Pi_T$.
For each $x$, restricting the integral to $\{f_\rho>0\}$ and using $\rho(dv)=f_\rho(v)\,\Pi_T(dv)$ on this set,
\[
 Q_{\Pi_T}(x)\ge\int_{\{f_\rho>0\}}P_{v,T}(x)\,\Pi_T(dv)=\int\frac{P_{v,T}(x)}{f_\rho(v)}\,\rho(dv).
\]
Jensen's inequality for the concave logarithm gives
\[
 \log Q_{\Pi_T}(x)\ge\int\log P_{v,T}(x)\,\rho(dv)-\int\log f_\rho(v)\,\rho(dv)
 =\int\log P_{v,T}(x)\,\rho(dv)-\KL(\rho\|\Pi_T).
\]
Subtracting this from $\log P_{\theta,T}(x)$,
\[
 \log\frac{P_{\theta,T}(x)}{Q_{\Pi_T}(x)}\le\int\log\frac{P_{\theta,T}(x)}{P_{v,T}(x)}\,\rho(dv)+\KL(\rho\|\Pi_T).
\]
Averaging over $x\sim P_{\theta,T}$ and exchanging the finite sum over $x$ with the integral over $v$ gives \eqref{eq:mixture-variation}. All integrals are finite: $\log P_{v,T}(x)\le0$, and by \eqref{eq:loglik-floor} its negative part is integrable because $\rho$ has finite first moments. In \Cref{sec:upper}, $\rho$ is a product of Gaussian laws and point masses, so both conditions hold.

\subsection{Proof of Lemma~\ref{lem:van-trees} (van Trees inequality)}\label{app:information}
We use the notation of \Cref{sec:information}: the design $(Z,b)$ is fixed, $\varphi$ is the joint density of $(\vartheta,W_{1:t-1})$, and $\zeta$ and $J_{t-1}$ are given by \eqref{eq:joint-score}.
\begin{proof}
The substitution $u=\pi v/(2r_j)$ gives $\int_{-r_j}^{r_j}\cos^2(\pi v/(2r_j))\,dv=r_j$, so $f_j$ is a probability density. Its derivative is
\[
 f_j'(v)=-\frac{\pi}{r_j^2}\cos\Bigl(\frac{\pi v}{2r_j}\Bigr)\sin\Bigl(\frac{\pi v}{2r_j}\Bigr),
\]
and both $f_j$ and $f_j'$ vanish at $v=\pm r_j$. Inside the interval, $(\log f_j)'(v)=-(\pi/r_j)\tan(\pi v/(2r_j))$. The score has mean $\int f_j'(v)\,dv=f_j(r_j)-f_j(-r_j)=0$, and its variance is
\[
 \int_{-r_j}^{r_j}\left(\frac{f'_j(v)}{f_j(v)}\right)^2f_j(v)\,dv
 =\frac{\pi^2}{r_j^3}\int_{-r_j}^{r_j}\sin^2\left(\frac{\pi v}{2r_j}\right)dv
 =\frac{\pi^2}{r_j^2}.
\]
For (b), expand $\zeta\zeta^\top$ using \eqref{eq:joint-score}. The coordinates of $\vartheta$ are independent and each prior score has mean zero, so by the computation of the prior score above, $\E[\nabla\log f\,(\nabla\log f)^\top]=\diag(\pi^2/r_j^2)=\pi^2R^{-2}$. Given $\vartheta$, the variables $W_s-p_s(\vartheta)$ have mean zero, so the cross terms between the prior score and the label terms have mean zero. Given $\vartheta$, the labels are independent, so $\E[(W_s-p_s)(W_{s'}-p_{s'})]=0$ for $s\ne s'$, while for $s=s'$ it equals $\E[p_s(1-p_s)]=\E[\psi''(b_s+z_s^\top\vartheta)]$. Summing these contributions gives (b).
For (a), fix $w\in\{0,1\}^{t-1}$ and a coordinate $j$. For fixed values of the other coordinates, integrate by parts in $\vartheta_j$ over $[-r_j,r_j]$:
\[
 \int_{-r_j}^{r_j}(\widehat g(w)-g(\vartheta))\,\partial_j\varphi(\vartheta,w)\,d\vartheta_j
 =\int_{-r_j}^{r_j}\partial_jg(\vartheta)\,\varphi(\vartheta,w)\,d\vartheta_j .
\]
The boundary term vanishes because $\varphi(\vartheta,w)$ contains the factor $f_j(\vartheta_j)$, which is zero at $\pm r_j$. Integrating over the other coordinates, summing over $w$, and using $\partial_j\varphi=\zeta_j\varphi$, we obtain $\E[(\widehat g-g(\vartheta))\zeta]=\E\nabla g(\vartheta)=d$.
If $\E(\widehat g-g)^2=\infty$ there is nothing to prove. Otherwise, for every $c\in\R^m$, the identity $\E[(\widehat g-g(\vartheta))\zeta]=d$ and the Cauchy--Schwarz inequality give
\[
 (c^\top d)^2=\bigl(\E[(\widehat g-g)\,c^\top\zeta]\bigr)^2\le\E(\widehat g-g)^2\;c^\top J_{t-1}c .
\]
By (b), $J_{t-1}\succeq\pi^2R^{-2}$ is positive definite. Take $c=J_{t-1}^{-1}d$. If $d=0$, (a) is trivial; otherwise the inequality reads $(d^\top J_{t-1}^{-1}d)^2\le\E(\widehat g-g)^2\,d^\top J_{t-1}^{-1}d$, and dividing by $d^\top J_{t-1}^{-1}d>0$ proves (a).
\end{proof}

\subsection{Proof of Lemma~\ref{lem:ons-pathwise} (pathwise bound for the online Newton step)}\label{app:ons-pathwise}
\begin{proof}
For $u,v\in K$, both $x_t^\top u$ and $x_t^\top v$ lie in $[-B,B]$, where $\psi''\ge\kappa$. Taylor's theorem for $\psi$ along the segment between them gives
\begin{equation}
 \ell_t(v)\ge\ell_t(u)+\nabla\ell_t(u)^\top(v-u)+\frac\kappa2\bigl(x_t^\top(v-u)\bigr)^2,
 \qquad\nabla\ell_t(u)=(\sig(x_t^\top u)-y_t)x_t . \label{eq:ons-taylor}
\end{equation}
Apply \eqref{eq:ons-taylor} with $u_t$ in place of $u$ and $u$ in place of $v$. Since $\nabla\ell_t(u_t)=g_t$ and $\abs{\sig(x_t^\top u_t)-y_t}\le1$, we have $(g_t^\top(u_t-u))^2\le(x_t^\top(u_t-u))^2$, and therefore
\begin{equation}
 \ell_t(u_t)-\ell_t(u)
 \le g_t^\top(u_t-u)-\frac\kappa2\bigl(g_t^\top(u_t-u)\bigr)^2 .
 \label{eq:ons-curvature}
\end{equation}
Fix $u\in K$. Writing $a_t-u=(a_t-u_{t+1})+(u_{t+1}-u)$ and expanding in the inner product of $H_t$,
\[
 \norm{u_{t+1}-u}_{H_t}^2=\norm{a_t-u}_{H_t}^2-\norm{u_{t+1}-a_t}_{H_t}^2+2(u_{t+1}-u)^\top H_t(u_{t+1}-a_t)\le\norm{a_t-u}_{H_t}^2+2\eps_t ,
\]
where the inequality uses \eqref{eq:certificate}. Expanding $\norm{a_t-u}_{H_t}^2=\norm{u_t-u}_{H_t}^2-2g_t^\top(u_t-u)+g_t^\top H_t^{-1}g_t$ and rearranging,
\begin{equation}
 g_t^\top(u_t-u)
 \le\frac12\bigl(\norm{u_t-u}_{H_t}^2-\norm{u_{t+1}-u}_{H_t}^2\bigr)
       +\frac12g_t^\top H_t^{-1}g_t+\eps_t. \label{eq:ons-projection}
\end{equation}
Since $H_t=H_{t-1}+\kappa g_tg_t^\top$, we have $\norm{u_t-u}_{H_t}^2=\norm{u_t-u}_{H_{t-1}}^2+\kappa(g_t^\top(u_t-u))^2$. Inserting \eqref{eq:ons-projection} into \eqref{eq:ons-curvature}, the term $\kappa(g_t^\top(u_t-u))^2/2$ cancels, and
\[
 \ell_t(u_t)-\ell_t(u)\le\frac12\bigl(\norm{u_t-u}_{H_{t-1}}^2-\norm{u_{t+1}-u}_{H_t}^2\bigr)+\frac12g_t^\top H_t^{-1}g_t+\eps_t .
\]
Summing over $t$, the first terms telescope. With $u_1=0$ and $H_0=I_m$,
\begin{equation}
 \begin{split}
 \sum_{t=1}^{T}(\ell_t(u_t)-\ell_t(u))
 &\le\frac12\norm u_2^2-\frac12\norm{u_{T+1}-u}_{H_T}^2
       +\frac12\sum_{t=1}^{T}g_t^\top H_t^{-1}g_t+\sum_{t=1}^T\eps_t\\
 &\le\frac12\norm u_2^2+\frac12\sum_{t=1}^{T}g_t^\top H_t^{-1}g_t+\sum_{t=1}^T\eps_t.
 \end{split}
 \label{eq:ons-telescope}
\end{equation}
Let $\gamma_t=\kappa g_t^\top H_{t-1}^{-1}g_t\ge0$. The Sherman--Morrison formula gives $H_t^{-1}g_t=H_{t-1}^{-1}g_t/(1+\gamma_t)$, so $\kappa g_t^\top H_t^{-1}g_t=\gamma_t/(1+\gamma_t)$. The matrix determinant lemma gives $\det H_t=(1+\gamma_t)\det H_{t-1}$. Since $x/(1+x)\le\log(1+x)$ for $x\ge0$ (both sides vanish at $x=0$, and the derivative $(1+x)^{-2}$ of the left side is at most the derivative $(1+x)^{-1}$ of the right side),
\[
 \frac12\sum_{t=1}^Tg_t^\top H_t^{-1}g_t\le\frac1{2\kappa}\sum_{t=1}^T\log\frac{\det H_t}{\det H_{t-1}}=\frac1{2\kappa}\log\det H_T .
\]
Finally, we have $\kappa g_tg_t^\top=\kappa(\sig(x_t^\top u_t)-y_t)^2x_tx_t^\top\preceq x_tx_t^\top$, because $\kappa\le1/4$ and the squared factor is at most one. Hence $H_T\preceq I_m+\sum_tx_tx_t^\top$. The determinant is monotone on positive definite matrices ordered by $\preceq$, so \eqref{eq:ons-telescope} and the determinant bound above give \eqref{eq:ons-determinant}.
\end{proof}

\section{Proofs of the Rates and Leading Constants of Section~\ref{sec:rates}}\label{app:rates}

\subsection{Proof of Corollary~\ref{cor:poly} (polynomial envelopes)}\label{app:poly}
Write $p=2s>2$ and $M=M_T=(A^2T)^{1/p}$, and let $f(x)=\log(1+x^{-p})$ for $x>0$.
For $j\le T$,
\[
 n_{T,j}r_j^2=(T-j+1)A^2j^{-p}=\Bigl(1-\frac{j-1}{T}\Bigr)\Bigl(\frac Mj\Bigr)^p,
\]
because $TA^2=M^p$. In particular $n_{T,j}r_j^2\le(M/j)^p$, so the term of lag $j$ is at most $f(j/M)$.
Assume $2\le M\le T/2$. For $j\le\lfloor M\rfloor$ we have $(j-1)/T<1/2$ and $M/j\ge1$, so $n_{T,j}r_j^2\ge1/2$ and the term of lag $j$ is at least $\log(3/2)$. Since $\lfloor M\rfloor\ge M-1\ge M/2$,
\begin{equation}
 \Gamma_T(r)\ge\lfloor M\rfloor\log(3/2)\ge\frac M2\log(3/2).
 \label{eq:poly-finite-lower}
\end{equation}
The function $f$ is positive and decreasing, so for every $j\ge1$
\[
 f(j/M)\le M\int_{(j-1)/M}^{j/M}f(x)\,dx .
\] Summing over $j$ and using $n_{T,j}r_j^2\le(M/j)^p$,
\begin{equation}
 \Gamma_T(r)\le\sum_{j=1}^{\infty}f(j/M)\le M\int_0^\infty f(x)\,dx.
 \label{eq:poly-finite-upper}
\end{equation}
The integral is finite: near zero $f(x)=p\log(1/x)+\log(1+x^p)$, which is integrable, and $f(x)\le x^{-p}$ for all $x$.
Integrating by parts, with $xf(x)\to0$ as $x\to0$ and as $x\to\infty$ and $f'(x)=-p/(x(1+x^p))$,
\begin{equation}
 \int_0^\infty f(x)\,dx=-\int_0^\infty xf'(x)\,dx
 =p\int_0^\infty\frac{dx}{1+x^p}
 =\int_0^\infty\frac{u^{1/p-1}}{1+u}\,du
 =\frac\pi{\sin(\pi/p)}. \label{eq:poly-integral}
\end{equation}
The third equality substitutes $u=x^p$. The last is Euler's integral $\int_0^\infty u^{c-1}(1+u)^{-1}du=\pi/\sin(\pi c)$ for $c\in(0,1)$, used with $c=1/p$.
Together, \eqref{eq:poly-finite-upper} and \eqref{eq:poly-integral} are \eqref{eq:poly-gamma-upper}, and they hold for every $T\ge1$ because the proof of \eqref{eq:poly-finite-upper} uses only $n_{T,j}\le T$. By \Cref{thm:main}, \eqref{eq:poly-finite-lower} and \eqref{eq:poly-gamma-upper},
\[
 \frac{c_*\log(3/2)}{2\max\{1,B^2\}}M\le\cR_T(r)\le\frac{\pi}{2\sin(\pi/(2s))}M ,
\]
and neither constant depends on $A$ or $T$.
Fix $A$ and $s$ and let $T\to\infty$. Then $M\to\infty$ and $M^2/T=A^{2/s}T^{1/s-1}\to0$. By \eqref{eq:poly-finite-upper} and \eqref{eq:poly-integral}, $\limsup_T\Gamma_T(r)/M\le\pi/\sin(\pi/p)$. For the matching lower bound, fix $0<x_0<x_1$ and take $T$ so large that $x_1M\le T/2$. For $j\le x_1M$, write $n_{T,j}r_j^2=c_j(M/j)^p$ with $c_j=1-(j-1)/T\in[1-x_1M/T,1]$. For $c\in(0,1]$ and $z\ge0$ we have $1+cz\ge c(1+z)$, so
\[
 \log(1+n_{T,j}r_j^2)\ge f(j/M)+\log(1-x_1M/T).
\]
Summing over $x_0M<j\le x_1M$ and dividing by $M$,
\[
 \frac{\Gamma_T(r)}{M}\ge\frac1M\sum_{x_0M<j\le x_1M}f(j/M)-x_1\abs{\log(1-x_1M/T)} .
\]
The first term on the right is a Riemann sum for the continuous function $f$ on $[x_0,x_1]$ and converges to $\int_{x_0}^{x_1}f(x)\,dx$. The second term is at most $2x_1^2M/T\to0$, by $\abs{\log(1-u)}\le2u$ for $u\in[0,1/2]$. Hence $\liminf_T\Gamma_T(r)/M\ge\int_{x_0}^{x_1}f(x)\,dx$. Letting $x_0\to0$ and $x_1\to\infty$ and using \eqref{eq:poly-integral} proves $\Gamma_T(r)\sim\pi M/\sin(\pi/(2s))$. The statement $\cR_T(r)=\Theta_{B,s}(M_T)$ follows from \Cref{thm:main}.

\subsection{Proof of Proposition~\ref{prop:window} (dimension-only window bound)}\label{app:window}
Let $h\ge1$ and run the online Newton step of \Cref{lem:ons-pathwise}, with exact projections ($\eps_t=0$), the unscaled features $x_t=(U_{t+1-j}\ind\{t\ge j\})_{j\le h}$ and the box $K=\prod_{j\le h}[-r_j,r_j]$, on which every logit lies in $[-B,B]$. For $\theta\in\Theta_r$, the comparator $u=(\theta_1,\ldots,\theta_h)\in K$ satisfies $\norm u_2^2\le(\sum_j\abs{\theta_j})^2\le B^2$. The diagonal entries of $\sum_tx_tx_t^\top$ are $n_{T,j}\le T$, so Hadamard's inequality gives $\log\det(I_h+\sum_tx_tx_t^\top)\le h\log(1+T)$. By \eqref{eq:ons-determinant} the pathwise regret against $u$ is at most
\[
 \frac{B^2}{2}+\frac{h}{2\kappa_B}\log(1+T)\le C_Bh\log(1+T),
 \qquad C_B=\frac{B^2}{2\log2}+\frac1{2\kappa_B},
\]
where we used $h\ge1$ and $\log(1+T)\ge\log2$.
Apply \eqref{eq:regret-split} with $v=\theta^S$, $S=\{1,\ldots,h\}$. The predictor gives each new input its uniform law, and $\sig(x_t^\top u)$ is the probability of $Y_{t+1}=1$ under $\theta^S$, so the last sum in \eqref{eq:regret-split} equals $\E_\theta\sum_t(\ell_t(u_t)-\ell_t(u))$, which is at most $C_Bh\log(1+T)$ by the bound above. The first term is $\KL(P_{\theta,T}\|P_{\theta^S,T})$. By \Cref{lem:pairwise} and $\theta_j^2\le r_j^2$, this gives \eqref{eq:window-benchmark}.
It remains to prove \eqref{eq:window-rate}. Let $r_j=Aj^{-p/2}$ with $p=2s$. For every $h\ge1$,
\[
 \sum_{j>h}n_{T,j}A^2j^{-p}\le TA^2\int_h^\infty x^{-p}\,dx=\frac{TA^2}{p-1}h^{1-p}.
\]
If $h\le T/4$, the lags $j=h+1,\ldots,2h$ satisfy $n_{T,j}\ge T-2h+1\ge T/2$ and $j^{-p}\ge(2h)^{-p}$, so
\[
 \sum_{j>h}n_{T,j}A^2j^{-p}\ge h\cdot\frac T2\cdot A^2(2h)^{-p}=2^{-p-1}TA^2h^{1-p}.
\]
Thus, for $1\le h\le T/4$ and $T\ge2$, the expression \eqref{eq:window-benchmark} lies between two positive multiples, depending on $B$ and $s$, of $\Phi(h)=h\log T+TA^2h^{1-p}$.
Let $h_*=(TA^2/\log T)^{1/p}$. Substituting $h=yh_*$ with $y>0$,
\[
 \Phi(yh_*)=(TA^2)^{1/p}(\log T)^{1-1/p}\bigl(y+y^{1-p}\bigr).
\]
The function $y\mapsto y+y^{1-p}$ is continuous on $(0,\infty)$, tends to infinity at both ends, and has a positive minimum $c_p$. Hence $\Phi(h)\ge c_p(TA^2)^{1/p}(\log T)^{1-1/p}$ for every real $h>0$. For $T$ large, $h_*\ge1$ and $2h_*\le T/4$; the integer $h=\lceil h_*\rceil$ has $y\in[1,2]$, so $\Phi(h)\le(2+1)(TA^2)^{1/p}(\log T)^{1-1/p}$. Together these prove \eqref{eq:window-rate}. For $h>T/4$, the first term of \eqref{eq:window-benchmark} alone is at least $C_B(T/4)\log T$, which is larger than the value at $\lceil h_*\rceil$ for large $T$. Hence the minimum of \eqref{eq:window-benchmark} over $h$ is also of order $T^{1/(2s)}(\log T)^{1-1/(2s)}$.

\subsection{Proof of Proposition~\ref{prop:resolution} (lower bound with leading factor one half)}\label{app:resolution}
Let $m=\abs S$. Give $\theta_j$, $j\in S$, independent uniform priors on $[-r_j/2,r_j/2]$, and set $\theta_j=0$ for $j\notin S$; this prior is supported in $\Theta_r$. Write $n_j=n_{T,j}$ and define the rescaled parameter and rows
\[
 \phi_j=\sqrt{n_j}\,\theta_j,\qquad w_t=\bigl(Z_{t,j}/\sqrt{n_j}\bigr)_{j\in S},
\]
with $Z$ from \eqref{eq:lag-design}. Then $\phi$ is uniform on $K=\prod_{j\in S}[-\sqrt{n_j}r_j/2,\sqrt{n_j}r_j/2]$, the logit of $Y_{t+1}$ is $w_t^\top\phi$, and for every $v\in K$, $\abs{w_t^\top v}\le\sum_{j\in S}\abs{Z_{t,j}}r_j/2\le B/2$.
Let $\widehat G=\sum_{t=1}^{T}w_tw_t^\top$, so that $\widehat G_{jk}=G_{jk}/\sqrt{n_jn_k}$. By \eqref{eq:gram-moments}, $\widehat G_{jj}=1$, and
\[
 \E\norm{\widehat G-I_m}_F^2=\sum_{j\ne k\in S}\frac{\E G_{jk}^2}{n_jn_k}=\sum_{j\ne k\in S}\frac{\min\{n_j,n_k\}}{n_jn_k}=v_S .
\]
Let $\mathcal F=\{\norm{\widehat G-I_m}_{\mathrm{op}}\le1/2\}$. Markov's inequality gives $\Pp(\mathcal F^c)\le4v_S$, and on $\mathcal F$ we have $\widehat G\succeq\frac12I_m$.
Fix an input sequence $u$ in $\mathcal F$. Because $\phi$ is independent of the inputs, $\phi$ is still uniform on $K$ given $U_{1:T}=u$. Let $\widehat\phi$ be a minimizer over $K$ of the negative log-likelihood
\[
 \mathcal L(v)=\sum_{t=1}^{T}\bigl\{\psi(w_t^\top v)-Y_{t+1}w_t^\top v\bigr\};
\]
it exists because $\mathcal L$ is continuous and $K$ is compact.
For $v\in K$, every logit lies in $[-B/2,B/2]$, so
\begin{equation}
 \nabla^2\mathcal L(v)=\sum_t\psi''(w_t^\top v)\,w_tw_t^\top\succeq\kappa_{B/2}\widehat G\succeq\frac{\kappa_{B/2}}2I_m . \label{eq:mle-curvature}
\end{equation}
Let $\Delta=\widehat\phi-\phi$. The segment from $\phi$ to $\widehat\phi$ lies in the convex set $K$, so Taylor's theorem and \eqref{eq:mle-curvature} give $\mathcal L(\widehat\phi)\ge\mathcal L(\phi)+\nabla\mathcal L(\phi)^\top\Delta+\frac{\kappa_{B/2}}4\norm\Delta_2^2$. Since $\mathcal L(\widehat\phi)\le\mathcal L(\phi)$, the Cauchy--Schwarz inequality gives
\[
 \frac{\kappa_{B/2}}4\norm\Delta_2^2\le-\nabla\mathcal L(\phi)^\top\Delta\le\norm{\nabla\mathcal L(\phi)}_2\norm\Delta_2,
 \qquad\text{hence}\qquad
 \norm\Delta_2\le\frac{4}{\kappa_{B/2}}\norm{\nabla\mathcal L(\phi)}_2 .
\]
We have $\nabla\mathcal L(\phi)=\sum_t(\sig(w_t^\top\phi)-Y_{t+1})w_t$. Given $\phi$ and the inputs, the summands are independent with mean zero, and $\E(\sig(w_t^\top\phi)-Y_{t+1})^2\le1/4$. Hence
\[
 \E\bigl[\norm{\nabla\mathcal L(\phi)}_2^2\bigm|\phi,u\bigr]\le\frac14\sum_t\norm{w_t}_2^2=\frac14\tr\widehat G=\frac m4 ,
\]
and by the bound on $\norm\Delta_2$, $\E[\norm{\widehat\phi-\phi}_2^2\mid u]\le16m/(4\kappa_{B/2}^2)=4m/\kappa_{B/2}^2$.
Let $\dent$ denote differential entropy, conditional on $U_{1:T}=u$. The uniform law on $K$ has $\dent(\phi)=\sum_{j\in S}\log(\sqrt{n_j}r_j)=\frac12\sum_{j\in S}\log(n_jr_j^2)$. For a random vector $X\in\R^m$ with density and $\E\norm X_2^2=m\sigma^2$, the Gaussian maximum-entropy bound \citep{coverthomas2006} follows from nonnegativity of the relative entropy from the law of $X$ to $N(0,\sigma^2I_m)$ which gives $0\le-\dent(X)+\frac m2+\frac m2\log(2\pi\sigma^2)$, that is, $\dent(X)\le\frac m2\log(2\pi e\,\E\norm X_2^2/m)$. Since $\widehat\phi$ is a function of the marks,
\[
 \dent(\phi\mid Y_{2:T+1})=\dent(\phi-\widehat\phi\mid Y_{2:T+1})\le\dent(\phi-\widehat\phi)
 \le\frac m2\log\left(\frac{2\pi e}{m}\cdot\frac{4m}{\kappa_{B/2}^2}\right)=\eta_Bm,
\]
where the equality is translation invariance, the first inequality holds because conditioning does not increase entropy, and the last step uses the bound on $\E[\norm{\widehat\phi-\phi}_2^2\mid u]$.
For $u\in\mathcal F$, these entropy bounds give
\[
 I(\phi;Y_{2:T+1}\mid U_{1:T}=u)=\dent(\phi)-\dent(\phi\mid Y_{2:T+1})\ge\frac12\sum_{j\in S}\log(n_jr_j^2)-\eta_Bm .
\]
For $u\notin\mathcal F$, the mutual information is nonnegative. Averaging over the inputs and using $\Pp(\mathcal F)\ge1-4v_S$,
\[
 I(\phi;Y_{2:T+1}\mid U_{1:T})\ge\Pp(\mathcal F)\Bigl[\frac12\sum_{j\in S}\log(n_jr_j^2)-\eta_Bm\Bigr]_+
 \ge(1-4v_S)_+\Bigl[\frac12\sum_{j\in S}\log(n_jr_j^2)-\eta_Bm\Bigr]_+ .
\]
The map from $\theta$ to $\phi$ is a bijection on the support of the prior, so the left side equals $I(\theta;Y_{2:T+1}\mid U_{1:T})$. By \eqref{eq:conditional-information} and \eqref{eq:radius-lower}, this is a lower bound on $\cR_T(r)$, which proves \eqref{eq:resolution-lower}.

\subsection{Proof of Theorem~\ref{thm:leading} (the leading constant)}\label{app:leading}
Abbreviate $\Gamma=\Gamma_T(r)$, $d=d_T(r)$ and $s_j=s_{T,j}$, and fix $\tau\ge2$.
Let $S_\tau=\{j\le T:s_j\ge\tau\}$. Each $j\in S_\tau$ has $\min\{1,s_j\}=1$, so $\abs{S_\tau}\le d$. For $j\in S_\tau$ we have $n_{T,j}\ge\tau/r_j^2$, hence $1/\sqrt{n_{T,j}}\le r_j/\sqrt\tau$, and for $j,k\in S_\tau$
\[
 \frac1{\max\{n_{T,j},n_{T,k}\}}\le\frac1{\sqrt{n_{T,j}n_{T,k}}}\le\frac{r_jr_k}{\tau}.
\]
Summing over $j\ne k$ in $S_\tau$ and using $\sum_{j\ne k}r_jr_k\le(\sum_jr_j)^2\le B^2$ gives $v_{S_\tau}\le B^2/\tau$.
For $j\le T$ with $j\notin S_\tau$ we show $\log(1+s_j)\le\log(1+\tau)\min\{1,s_j\}$. If $s_j<1$, this follows from $\log(1+s_j)\le s_j$ and $\log(1+\tau)\ge\log3>1$. If $1\le s_j<\tau$, it follows from $\log(1+s_j)<\log(1+\tau)$. For $j\in S_\tau$ we have $\log(1+s_j)-\log s_j=\log(1+1/s_j)\le\log2\le\log(1+\tau)=\log(1+\tau)\min\{1,s_j\}$. Adding the two cases over $j\le T$,
\begin{equation}
 \sum_{j\in S_\tau}\log s_j\ge\Gamma-\log(1+\tau)\,d. \label{eq:threshold-logs}
\end{equation}
We now prove \eqref{eq:threshold-lower}. If $S_\tau$ is empty, \eqref{eq:threshold-logs} gives $\Gamma\le\log(1+\tau)d$, the bracket in \eqref{eq:threshold-lower} is nonpositive, and \eqref{eq:threshold-lower} holds because regret is nonnegative. Otherwise $S_\tau\subseteq\{j\le T:r_j>0\}$ is nonempty, and \Cref{prop:resolution} applies with $m=\abs{S_\tau}\le d$. By the bound $v_{S_\tau}\le B^2/\tau$ and \eqref{eq:threshold-logs},
\[
 \cR_T(r)\ge(1-4v_{S_\tau})_+\Bigl[\frac12\sum_{j\in S_\tau}\log s_j-\eta_B\abs{S_\tau}\Bigr]_+
 \ge\Bigl(1-\frac{4B^2}{\tau}\Bigr)_+\Bigl[\frac12\Gamma-\Bigl(\eta_B+\frac12\log(1+\tau)\Bigr)d\Bigr]_+ .
\]
We turn to \eqref{eq:leading-general} and first show that the condition $d_T=o(\Gamma_T)$ is sufficient. Write $\Gamma_T=\Gamma_T(r^{(T)})$ and $d_T=d_T(r^{(T)})$. For large $T$, $\Gamma_T>0$, hence some $s_{T,j}>0$ and $d_T>0$; let $q_T=\Gamma_T/d_T$, which tends to infinity by assumption. Apply \eqref{eq:threshold-lower} with $\tau_T=\max\{2,q_T\}$, which equals $q_T$ for large $T$. Then $4B^2/\tau_T\to0$, $\eta_Bd_T/\Gamma_T=\eta_B/q_T\to0$ and $\log(1+\tau_T)d_T/\Gamma_T=\log(1+q_T)/q_T\to0$. Hence $\cR_T(r^{(T)})\ge(\frac12-o(1))\Gamma_T$. The upper bound in \eqref{eq:main} gives $\cR_T(r^{(T)})\le\frac12\Gamma_T$, so $\cR_T(r^{(T)})\sim\frac12\Gamma_T$.
Conversely, suppose $\cR_T(r^{(T)})/\Gamma_T\to\frac12$. By \Cref{lem:gap},
\[
 0\le\delta\,\frac{d_T}{\Gamma_T}\le\frac12-\frac{\cR_T(r^{(T)})}{\Gamma_T}\longrightarrow0,
\]
so $d_T=o(\Gamma_T)$. The two implications prove \eqref{eq:leading-general}.

\subsection{Proof of Corollary~\ref{cor:leading-examples} (finite support and stretched-exponential envelopes)}\label{app:leading-examples}
For a finitely supported envelope, let $S$ be the set of $j$ with $r_j>0$ and $k=\abs S$. Let $T>\max S$. Then $T-\max S<n_{T,j}\le T$ for $j\in S$, so $\Gamma_T(r)=k\log T+O_r(1)$, and $d_T(r)\le k$. The upper bound in \eqref{eq:main} gives $\cR_T(r)\le\frac12\Gamma_T(r)=\frac k2\log T+O_r(1)$. For the lower bound, apply \Cref{prop:resolution} to $S$. Then $v_S\le k(k-1)/(T-\max S+1)=O_r(1/T)$ and $\frac12\sum_{j\in S}\log(n_{T,j}r_j^2)-\eta_Bk=\frac k2\log T+O_{B,r}(1)$, so
\[
 \cR_T(r)\ge(1-O_r(1/T))\Bigl(\frac k2\log T+O_{B,r}(1)\Bigr)=\frac k2\log T+O_{B,r}(1).
\]
For stretched-exponential envelopes, let $r_j=Ae^{-\alpha j^\beta}$ and write $L=L_T=\log(A^2T)$, $\mu=(L/(2\alpha))^{1/\beta}$ and $N=\lfloor \mu\rfloor$. As $T\to\infty$, $L\to\infty$, $\mu\to\infty$ and $\mu^2/T\to0$; we take $T$ so large that $L\ge1$, $\mu\ge1$ and $\mu\le T/2$. Define
\[
 S_L=\sum_{j=1}^{N}(L-2\alpha j^\beta),
\]
whose terms are nonnegative because $2\alpha j^\beta\le2\alpha \mu^\beta=L$ for $j\le N$.
Comparing the increasing function $x^\beta$ with its integrals over $[j-1,j]$ and $[j,j+1]$ gives $\sum_{j=1}^Nj^\beta=N^{\beta+1}/(\beta+1)+O_\beta(N^\beta+1)$. Since $\mu-1<N\le \mu$ and $2\alpha \mu^\beta=L$, we have $NL=\mu L+O(L)$ and $2\alpha N^{\beta+1}=2\alpha \mu^{\beta+1}+O_\beta(L)=L\mu+O_\beta(L)$. Hence
\begin{equation}
 S_L=L\mu-\frac{L\mu}{\beta+1}+O_{\alpha,\beta}(L+1)=\frac{\beta}{\beta+1}L\mu\,(1+o(1)).
 \label{eq:SL-equivalent}
\end{equation}
Since $n_{T,j}\le T$, we have $n_{T,j}r_j^2\le e^{L-2\alpha j^\beta}$. For $j\le N$ the exponent $y=L-2\alpha j^\beta$ is nonnegative and $\log(1+e^y)\le y+\log2$. For $j>N$ we use $\log(1+e^y)\le e^y$. The function $x\mapsto e^{L-2\alpha x^\beta}$ is decreasing, $N+1>\mu$, and its value at $\mu$ is one, so
\[
 \sum_{j>N}e^{L-2\alpha j^\beta}\le1+\int_\mu^\infty e^{L-2\alpha x^\beta}dx
 =1+\mu\int_1^\infty e^{-L(u^\beta-1)}du\le1+c_\beta \mu,
\]
where the equality substitutes $x=\mu u$, the last inequality uses $L\ge1$, and $c_\beta=\int_1^\infty e^{-(u^\beta-1)}du<\infty$. Therefore $\Gamma_T(r)\le S_L+N\log2+1+c_\beta \mu=S_L+O_\beta(\mu)$.
Using $\log(1+x)\ge\log x$ for the lags $j\le N$,
\[
 \Gamma_T(r)\ge\sum_{j\le N}\log(n_{T,j}r_j^2)=S_L+\sum_{j\le N}\log\Bigl(1-\frac{j-1}{T}\Bigr).
\]
Each $(j-1)/T$ is at most $\mu/T\le1/2$, and $\abs{\log(1-u)}\le2u$ on $[0,1/2]$, so the last sum has absolute value at most $2\sum_{j\le N}(j-1)/T\le \mu^2/T=o(1)$. With the upper bound $\Gamma_T(r)\le S_L+O_\beta(\mu)$ and \eqref{eq:SL-equivalent},
\begin{equation}
 \Gamma_T(r)=S_L+O_\beta(\mu)\sim\frac{\beta}{\beta+1}L\mu.
 \label{eq:stretched-Gamma}
\end{equation}
The same split as in the upper bound on $\Gamma_T(r)$ gives $d_T(r)\le N+\sum_{j>N}e^{L-2\alpha j^\beta}\le N+1+c_\beta\mu=O_\beta(\mu)$, while $\Gamma_T(r)\sim\frac{\beta}{\beta+1}L\mu$ by \eqref{eq:stretched-Gamma}. Since $L\to\infty$, $d_T(r)/\Gamma_T(r)\to0$, and \eqref{eq:leading-general} with the fixed envelope $r$ gives
\[
 \cR_T(r)\sim\frac12\Gamma_T(r)\sim\frac{\beta}{2(\beta+1)}L\mu=\frac{\beta}{2(\beta+1)}L\Bigl(\frac{L}{2\alpha}\Bigr)^{1/\beta},
\]
which is \eqref{eq:stretched-leading}. For $\beta=1$ the constant is $\frac14\cdot\frac1{2\alpha}=\frac1{8\alpha}$, which gives \eqref{eq:exponential-leading}.

\section{Proofs of Theorem~\ref{thm:ons} and Lemma~\ref{lem:projection-cost}}\label{app:algorithm}

\subsection{Proof of Theorem~\ref{thm:ons}}\label{app:ons}
\begin{proof}[Proof of \Cref{thm:ons}]
Let $m=\abs S$.
If $m>0$, apply \Cref{lem:ons-pathwise} (\Cref{app:ons-pathwise}) with $K=[-1,1]^m$, the features \eqref{eq:ons-features}, $y_t=Y_{t+1}$ and the tolerances $\eps_t$. Every $u\in K$ satisfies $\abs{x_t^\top u}\le\sum_{j\in S}r_j\le B$ and $\norm u_2^2\le m$.
The $j$-th diagonal entry of $\sum_tx_tx_t^\top$ is $\sum_{t=1}^{T}r_j^2U_{t+1-j}^2=n_{T,j}r_j^2$, because $U_{t+1-j}^2=1$ for $t\ge j$ and $U_{t+1-j}=0$ for $t<j$. Hadamard's inequality bounds the determinant of a positive definite matrix by the product of its diagonal entries, so $\log\det(I_m+\sum_tx_tx_t^\top)\le\sum_{j\in S}\log(1+n_{T,j}r_j^2)$. With \eqref{eq:ons-determinant}, for every input and mark sequence and every $u\in K$,
\begin{equation}
 \sum_{t\le T}(\ell_t(u_t)-\ell_t(u))
 \le D_{r,S}+\sum_{t\le T}\eps_t,\qquad D_{r,S}\defeq\frac m2+\frac1{2\kappa_B}\sum_{j\in S}\log(1+n_{T,j}r_j^2).
 \label{eq:pathwise-ons}
\end{equation}
Fix $\theta\in\Theta_r$, let $u_j=\theta_j/r_j\in[-1,1]$ for $j\in S$, and apply \eqref{eq:regret-split} with $v=\theta^S$. Then $\sig(x_t^\top u)$ is the probability of $Y_{t+1}=1$ under $\theta^S$. The predictor gives each new input its uniform law, so the last sum in \eqref{eq:regret-split} equals $\E_\theta\sum_{t\le T}(\ell_t(u_t)-\ell_t(u))$, which is at most $D_{r,S}+\sum_t\eps_t$ by \eqref{eq:pathwise-ons}. \Cref{lem:pairwise} bounds the first term, $\KL(P_{\theta,T}\|P_{\theta^S,T})$, by $\frac18\sum_{j\notin S}n_{T,j}\theta_j^2$. This proves \eqref{eq:ons-oracle}. If $m=0$, the predictor is the source with $\theta=0$ and computes no projections; its regret is $\KL(P_{\theta,T}\|P_{0,T})$, and \Cref{lem:pairwise} gives \eqref{eq:ons-oracle} without the terms $D_{r,S}$ and $\sum_t\eps_t$.
It remains to prove \eqref{eq:ons-final}. Write $s_j=n_{T,j}r_j^2$. For $j\in S_T$, $s_j\ge1$ gives $\log(1+s_j)\ge\log2$, so $\frac12\le\log(1+s_j)/(2\log2)$, and the $j$-th contributions to $D_{r,S}$ sum to at most $((2\kappa_B)^{-1}+(2\log2)^{-1})\log(1+s_j)$. For $j\le T$ with $j\notin S_T$, either $r_j=0$ and the term vanishes, or $s_j<1$, and then $\frac18n_{T,j}\theta_j^2\le s_j/8\le\log(1+s_j)/(8\log2)$ by the inequality $\log(1+s)\ge s\log2$ on $[0,1]$ proved before \eqref{eq:d-gamma}. If $S_T$ is nonempty, then $\Gamma_T(r)\ge\log2$, so $\sum_t\eps_t\le1\le\Gamma_T(r)/\log2$; if $S_T$ is empty, no projections are computed. Adding the three bounds gives \eqref{eq:ons-final}, since $(2\kappa_B)^{-1}+(2\log2)^{-1}+(\log2)^{-1}=C_B^{\mathrm{ONS}}$.
\end{proof}

\subsection{Proof of Lemma~\ref{lem:projection-cost} (cost of an approximate projection)}\label{app:projection}
Fix a round $t$ with $m>0$ and write $H=H_t$, $a=a_t$, $K=[-1,1]^m$ and $f(v)=\frac12\norm{v-a}_H^2$, so that $\nabla f(v)=H(v-a)$ and the left side of \eqref{eq:certificate} is $c(v)=\sup_{u\in K}\nabla f(v)^\top(v-u)$. Let $f^*=\min_Kf$ and $\nu=1+\kappa_BTB^2$.
We have $I_m\preceq H\preceq \nu I_m$. The lower bound holds because $H_0=I_m$. For the upper bound, $\norm{g_s}_2^2\le\norm{x_s}_2^2\le\sum_{j\in S}r_j^2\le B^2$, so the largest eigenvalue of $H=I_m+\kappa_B\sum_{s\le t}g_sg_s^\top$ is at most $1+\kappa_BtB^2\le \nu$.
For $v\in K$ let $v^+=\Pi_K(v-\nabla f(v)/\nu)$, where $\Pi_K$ is the Euclidean projection onto $K$, computed coordinatewise by clipping to $[-1,1]$. Equivalently, $v^+$ minimizes $\phi(w)=\nabla f(v)^\top(w-v)+\frac \nu2\norm{w-v}_2^2$ over $w\in K$, so its optimality condition is
\begin{equation}
 \bigl(\nabla f(v)+\nu(v^+-v)\bigr)^\top(u-v^+)\ge0\qquad\text{for all }u\in K. \label{eq:pg-optimality}
\end{equation}
Since $H\preceq \nu I_m$, $f(w)\le f(v)+\nabla f(v)^\top(w-v)+\frac \nu2\norm{w-v}_2^2$ for all $w$.
Let $v^*$ minimize $f$ on $K$ and take $w=(1-\nu^{-1})v+\nu^{-1}v^*\in K$. Since $v^+$ minimizes $\phi$ over $K$ and $f(w)\le f(v)+\phi(w)$ for all $w$,
\[
 f(v^+)\le f(v)+\phi(w)=f(v)+\frac1\nu\nabla f(v)^\top(v^*-v)+\frac1{2\nu}\norm{v^*-v}_2^2 .
\]
Since $H\succeq I_m$, $f(v^*)\ge f(v)+\nabla f(v)^\top(v^*-v)+\frac12\norm{v^*-v}_2^2$. Inserting this bound on $\nabla f(v)^\top(v^*-v)$ gives $f(v^+)-f^*\le(1-\nu^{-1})(f(v)-f^*)$.
Taking $u=v$ in \eqref{eq:pg-optimality} gives $\nabla f(v)^\top(v-v^+)\ge \nu\norm{v^+-v}_2^2$, and with $f(v^+)\le f(v)+\phi(v^+)$, $f(v^+)\le f(v)-\frac \nu2\norm{v^+-v}_2^2$, so $\norm{v^+-v}_2^2\le2(f(v)-f^*)/\nu$. For $u\in K$, \eqref{eq:pg-optimality} gives $\nabla f(v)^\top(v^+-u)\le \nu(v^+-v)^\top(u-v^+)$, and $H\preceq \nu I_m$ gives $\norm{\nabla f(v^+)-\nabla f(v)}_2\le \nu\norm{v^+-v}_2$. Hence
\[
 \nabla f(v^+)^\top(v^+-u)\le2\nu\norm{v^+-v}_2\norm{v^+-u}_2\le4\nu\sqrt m\,\norm{v^+-v}_2 ,
\]
because points of $K$ are at distance at most $2\sqrt m$. Therefore $c(v^+)\le4\sqrt{2m\nu\,(f(v)-f^*)}$.
Start at $v_0=u_t\in K$ and let $v_{k+1}=v_k^+$. Then $f(v_0)-f^*\le f(u_t)=\frac12g_t^\top H^{-1}g_t\le\frac12B^2$, because $H\succeq I_m$. The contraction $f(v^+)-f^*\le(1-\nu^{-1})(f(v)-f^*)$ and the bound on $c(v^+)$ give $c(v_{k+1})\le4\sqrt{2m\nu}\,(1-\nu^{-1})^{k/2}B/\sqrt2\le4B\sqrt{m\nu}\,e^{-k/(2\nu)}$. Hence $c(v_{k+1})\le1/T$ as soon as $k\ge2\nu\log_+(4B\sqrt{m\nu}\,T)$, that is, after at most $1+2\nu\log_+(4B\sqrt{m\nu}\,T)$ iterations. Each iteration costs $O(m^2)$ operations for the product $H(v-a)$ and $O(m)$ for the clipping.

\section{Further Related Work}\label{app:related}
The redundancy framework of universal coding is the starting point: \citet{davisson1973} introduced minimax redundancy, \citet{haussler1997} proved a general minimax theorem for relative entropy, and \citet{merhavfeder1998} survey universal prediction. The identity between expected cumulative log-loss regret and relative entropy, and the lower bound by the mutual information under a prior, are standard consequences; \Cref{app:base-radius} specializes them to the present source.

For a regular $d$-dimensional parametric family, \citet{clarkebarron1990} show that the Bayesian redundancy is $(d/2)\log T+O(1)$. \Cref{cor:leading-examples} recovers this constant for finitely supported envelopes, and \Cref{thm:leading} shows that $\frac12$ is the leading constant if and only if $d_T(r)=o(\Gamma_T(r))$, including cases in which the number of relevant lags grows with $T$. Its proof uses a finite-sample likelihood argument (\Cref{prop:resolution}) and avoids an asymptotic normal approximation, whose error would have to be controlled in growing dimension.

In online logistic regression the features are arbitrary and regret is measured against the best fixed parameter in hindsight. \citet{kakadeng2005} prove Bayesian bounds in this setting, \citet{shamir2020} proves upper and lower bounds under norm constraints, including regimes in which the dimension grows with $T$, and \citet{foster2018} use improper prediction to improve the dependence on the norm of the comparator. \citet{jacquetshamirszpankowski2021} compute the minimax regret for categorical features, and \citet{drmotajacquetwuszpankowski2026} study phase transitions for logistic regression with large weights. These upper bounds apply to windows of lagged inputs. Our lower bound differs in two ways: it uses the distribution of the input stream, and it assigns each lag the contribution $\log(1+n_{T,j}r_j^2)$.

\citet{yangbarron1999} relate nonparametric minimax rates to metric entropy, \citet{zhang2006} proves information-complexity bounds for density estimation, and \citet{wu2022} obtain precise log-loss bounds for classes of experts through smooth truncated coverings. Our upper bound applies the localization behind these results to each coordinate separately.

\citet{fokianostruquet2019} study categorical time series with covariates, including the construction of the process and its dependence properties. In our source the marks are conditionally independent given the inputs and the parameter, and the lower bounds use this structure. \citet{hanjiangwu2024} study prediction of the next symbol after a training sequence for infinite-memory processes, including hidden Markov and renewal processes, and relate the risk to redundancy and a memory term. Their criterion is the loss of one prediction; ours is the sum of the excess losses of all predictions.

\end{document}